\documentclass[floatfix,reprint,nofootinbib,amsmath,amssymb,epsfig,pre,floats,letterpaper]{revtex4-1}

\usepackage{amsmath}
\usepackage{amssymb}
\usepackage{amsthm}
\usepackage{bm}
\usepackage{dcolumn}
\usepackage[english]{babel}
\usepackage{enumitem}
\usepackage{graphicx}
\usepackage{hyperref}
\usepackage{inconsolata}
\usepackage{listings}
\usepackage{xcolor}
\usepackage{footmisc}

\DeclareMathAlphabet{\mathpzc}{OT1}{pzc}{m}{it}

\colorlet{punct}{red!60!black}
\definecolor{background}{HTML}{EEEEEE}
\definecolor{delim}{RGB}{20,105,176}
\colorlet{numb}{magenta!60!black}

\lstdefinelanguage{json}{
    basicstyle=\fontsize{8}{10}\ttfamily,
    showstringspaces=false,
    breaklines=true,
    literate=
     *{0}{{{\color{numb}0}}}{1}
      {1}{{{\color{numb}1}}}{1}
      {2}{{{\color{numb}2}}}{1}
      {3}{{{\color{numb}3}}}{1}
      {4}{{{\color{numb}4}}}{1}
      {5}{{{\color{numb}5}}}{1}
      {6}{{{\color{numb}6}}}{1}
      {7}{{{\color{numb}7}}}{1}
      {8}{{{\color{numb}8}}}{1}
      {9}{{{\color{numb}9}}}{1}
      {:}{{{\color{punct}{:}}}}{1}
      {,}{{{\color{punct}{,}}}}{1}
      {\{}{{{\color{delim}{\{}}}}{1}
      {\}}{{{\color{delim}{\}}}}}{1}
      {[}{{{\color{delim}{[}}}}{1}
      {]}{{{\color{delim}{]}}}}{1},
}

\newcommand{\beq}{\begin{equation}}
\newcommand{\eeq}{\end{equation}}

\newtheorem{theorem}{Theorem}

\theoremstyle{definition}

\theoremstyle{definition}

\theoremstyle{definition}

\begin{document}

\title{Decentralized Common Knowledge Oracles}

\author{Austin K. Williams}\email[]{austin.williams@onewayfunction.com}
\author{Jack Peterson}

\date{\today}

\begin{abstract}
We define and analyze three mechanisms for getting common knowledge, \emph{a posteriori} truths about the world onto a blockchain in a decentralized setting. We show that, when a reasonable economic condition is met, these mechanisms are individually rational, incentive compatible, and decide the true outcome of valid oracle queries in both the non-cooperative and cooperative settings. These mechanisms are based upon repeated games with two classes of players: \textit{queriers} who desire to get common knowledge truths onto the blockchain and a pool of \textit{reporters} who posses such common knowledge. Presented with a new oracle query, reporters have an opportunity to report the truth in return for a fee provided by the querier. During subsequent oracle queries, the querier has an opportunity to punish any reporters who did not report truthfully during previous rounds. While the set of reporters has the power to cause the oracle to lie, they are incentivized not to do so. 
\end{abstract}

\maketitle

\section{Introduction}

\subsection{Background}
In order for smart contracts to condition their execution on the state of the world, they need access to information about the world. While smart contracts can verify \emph{a priori} claims with mathematical or cryptographic certainty, they cannot independently verify \emph{a posteriori} claims about the world with the same assurances. As a matter of epistemological necessity, smart contracts which condition their behavior on \emph{a posteriori} knowledge must rely on trusted oracles to provide that knowledge. As a result, we can trust these smart contracts only if we can trust their oracles.

With no possibility of mathematical or cryptographic verification of \emph{a posteriori} claims about the world, we instead look to \emph{economic incentives} when considering whether to trust an oracle. We require that the cost (to the oracle operators) of lying be greater than the benefit. More specifically, we require that truth-telling be incentive compatible. We also want the operation of the oracle to have a non-negative expected return for the operators. That is, we require that the operation of the oracle be individually rational. Finally, we want the oracle to be decentralized in order to avoid both censorship and a single point of failure.

A common approach to designing such an oracle is to create a coordination game in which individual human players are presented with an oracle query and are asked to report the correct outcome by staking some tokens~\cite{Buterin_2014,Sztorc_2014}. The oracle outputs whichever outcome received the most stake as the ``winning outcome", and then players are rewarded if and only if they staked in agreement with the winning outcome. The hope in these ``Schelling scheme" approaches is that the truth will act as the Schelling point of a coordination game, which would result in the oracle returning the true outcome to the oracle query. Although these approaches are appealing because they are easy to implement and have a high degree of social scalability, they have serious drawbacks that make their real-world success unlikely.

First, Schelling points themselves are an informal solution concept used in the context of coordination games in which pre-play communication is incomplete or impossible, and in bargaining games in which players cannot make binding agreements~\cite{Schelling_1960}. However, in the open blockchain setting, players can freely engage in pre-play communication (via Reddit, Twitter, email, etc) and make binding agreements (for example, via smart contracts). So the Schelling point solution concept is not one that is known to be applicable in the types of strategic settings in which blockchain oracles operate. 

Second, the coordination game approach to oracle design can be incentive compatible only in the non-cooperative model. As soon as players are able to make binding agreements, they can perform bribing attacks and form coalitions that are large enough to make the oracle lie without members of the coalition receiving any penalty~\cite{Buterin_2015}. If players can form coalitions, we cannot rely on truth-telling being incentive compatible for mechanisms following the coordination game paradigm. The root of the problem is that players are rewarded for agreeing with the majority \emph{whether or not} the majority tells the truth.

For these reasons we think it is unlikely that the coordination-game approach to decentralized oracle design will work in practice. We desire an oracle that is incentive compatible in the cooperative model -- where players can engage in pre-play communication and make binding agreements.

In addition to being incentive compatible in isolation, one must also consider the trustworthiness of oracles in the presence of ``extraneous incentives": truth-telling must be incentive compatible when the output of the oracle is consumed to control the irreversible payout of large amounts of cryptocurrency. For example, every bet placed on a decentralized betting platform increases the gross incentive to make the platform's oracle lie. Many existing blockchain oracle designs do not explicitly consider the incentives introduced by the consumption of the oracle's output when analyzing the incentive compatibility of their mechanisms. (With the exception of Augur's oracle design~\cite{Peterson_2018}, we are unaware of \textit{any} proposed blockchain oracles, centralized or decentralized, that explicitly quantify and address this risk.) However, explicit consideration of such extraneous incentives is crucial when considering the security of oracles to be deployed for real-world use. For a given oracle design, we can quantify this risk and state precisely how much extraneous incentive the design can handle before losing incentive compatibility.

\subsection{Our Approach}

In this paper, we describe a new approach to oracle design that does not follow the coordination-game paradigm. For the mechanisms in this paper, truth-telling is -- under certain reasonable economic conditions -- incentive compatible in both the non-cooperative and cooperative game-theoretic models. A large coalition may form that makes the oracle lie, but members of such a coalition are severely penalized. Indeed, they are penalized more than they might gain from causing the oracle to lie. This stands in contrast to mechanisms based on coordination games, which \emph{reward} such large coalitions for lying.

At a high level, we begin with a separation of concerns. We create a distinction between those who want to get common knowledge \textit{a posteriori} information on-chain and those who possess such common knowledge. We refer to the former as \textit{queriers} and the latter as \textit{reporters}. For example, a trustless, decentralized platform for betting on the outcomes of elections would be a querier because it wants to get the true outcomes of elections on-chain in order to settle bets. A collection of humans that tells the platform which candidates have won would be a set of reporters.

We then create mechanisms based upon repeated games -- played by queriers and reporters -- wherein each stage game corresponds to an oracle query and the outcome of each stage game determines how the oracle responds to the query. We show that under certain economic conditions (which explicitly include any extraneous incentives introduced by any consumption of the oracle's output) there exists equilibrium behavior in these stage games that results in the oracle returning the true outcome of real world events. These first incentive compatibility proofs take place in the non-cooperative model and assume that the set of queriers and the set of reporters is disjoint.

We present three such mechanisms. Each mechanism is more complex than the last, but also has better scalability properties. We more closely examine the conditions under which our results hold and discuss the weaknesses of our approach. We also consider the cooperative game-theoretic model, where the queriers and reporters can form coalitions (or may even be the same people). We show that, with some additional conditions, our results hold in the cooperative model as well.

\section{Definitions}

\noindent \textbf{Definition} (\textit{Outcome space}).
For all events $\mathtt{E}$, an \emph{outcome space} of $\mathtt{E}$, denoted $\Omega_{\mathtt{E}}$, is a finite set of possible outcomes of $\mathtt{E}$. We require that every outcome space contain the special element \texttt{Invalid}, and that no outcome space contain the special element \texttt{Abstain}. When the event $\mathtt{E}$ is clear from context, we may drop the subscript and denote the outcome space $\Omega$.
\medskip

\noindent \textbf{Definition} (\textit{Oracles and queries}).
An \emph{oracle} is any algorithm that accepts as input an event $\mathtt{E}$, a corresponding outcome space $\Omega$, (and, optionally, some additional arguments) and outputs some $\omega \in \Omega$. A call to the oracle is referred to as a \emph{query}.
\medskip

\noindent \textbf{Definition} (\textit{Common knowledge}).
A proposition $P$ is said to be \emph{common knowledge} among a group of agents $G$ if all agents in $G$ know $P$, they all know that they all know $P$, they all know that they all know that they all know $P$, and so on, \textit{ad infinitum}~\cite{Aumann_1976}.
\medskip

The group of agents $G$ is the collection of all users who interact with the oracle. Informally, one may consider a proposition to be common knowledge if it can be quickly verified by any user with access to the World Wide Web.
\medskip

\noindent \textbf{Definition} (\textit{True outcome}).
For every oracle query with arguments $\mathtt{E}$ and $\Omega$, we define a unique outcome in $\Omega$ to be the \emph{true outcome} for the query. We denote such an outcome \texttt{True}, and it is defined as follows.  If there exists a unique outcome $ \omega \in \Omega \setminus \left\{ \texttt{Invalid} \right\}$ such that -- at the time of the oracle query -- it is common knowledge that the outcome of event $\mathtt{E}$ is $\omega$, then $\omega$ is the true outcome for the query. Otherwise, $\texttt{Invalid}$ is the true outcome for the query.
\medskip

\noindent \textbf{Definition} (\textit{False outcome}).
For every oracle query, every outcome in $\Omega$ that is not \texttt{True} is a \emph{false outcome}.
\medskip

It is important to note that simply corresponding to objective reality is not a sufficient condition for an outcome to be the \texttt{True} outcome for a query. The fact that the outcome corresponds to objective reality must also be common knowledge at the time of the query.
\medskip

\noindent \textbf{Definition} (\textit{Valid query}).
A query whose \texttt{True} outcome is not $\texttt{Invalid}$ is referred to a \emph{valid query}.
\medskip

Using this terminology, our objective for this paper is to construct a decentralized, incentive compatible, individually rational mechanism that decides the \texttt{True} outcome of valid queries.
\medskip

\noindent \textbf{Definition} (\textit{$\Omega$-partition}).
If $T$ is a finite set of tokens, $\mathtt{E}$ is an event, and $\Omega$ is an outcome space of $\mathtt{E}$, then an \textit{$\Omega$-partition of $T$} is an indexed family of $|\Omega|+1$ mutually disjoint subsets of $T$ (referred to as \textit{cells}) indexed by $\Omega \cup \{\texttt{Abstain}\}$ where the union of the cells is $T$. An $\Omega$-partition may be represented succinctly using the indexed-family notation $\textbf{C}=(C_{\omega})_{\omega \in \Omega \cup \left\{\texttt{Abstain} \right\} }$, or in expanded form via $\left\{ C_{\texttt{Abstain}}, C_{\omega_{1}}, \ldots, C_{\omega_{|\Omega|}} \right\}$.
\medskip

Colloquially, an $\Omega$-partition of a set of tokens is simply the separation of the tokens into ``piles'' (cells) that are labeled by the outcomes in $\Omega \cup \{\texttt{Abstain}\}$. Such partitions arise naturally in the context of voting with tokens. For example, suppose that everyone who owns at least one token in a set of tokens, $T$, is asked to cast a vote in favor of some outcome in $\Omega$. If we separate the tokens into cells according to how the owner of each token voted, the resulting partition would be an $\Omega$-partition of $T$. (Any tokens owned by someone who refused to vote is put into the pile labelled ``\texttt{Abstain}''.)

Next, we develop notation for a simple algorithm that asks a player to report the outcome of some event. The player's response, along with the set of tokens controlled by the player, are returned.
\medskip

\noindent \textbf{Definition} (\textit{Report}).
The algorithm $Report$ takes as input a tuple $(j, \mathtt{E},\Omega, T)$, where $T$ is a set of tokens, $j$ is the owner of at least one token in $T$, $\mathtt{E}$ is an event, and $\Omega$ is an outcome space of $\mathtt{E}$. The owner, $j$, is asked to report which element $\omega$ in $\Omega \cup \left\{ \texttt{Abstain} \right\}$ is \texttt{True}. If $j$ fails to respond then $\omega$ is understood to be \texttt{Abstain}. $Report$ returns the tuple $(\omega, R)$, where $R$ is the set of all tokens in $T$ that are owned by $j$.
\medskip

Next, we define an important algorithm referred to as the \emph{fork}. The fork is not an oracle, but will be used as an important subroutine in the oracles we construct in this paper. In brief, the fork is the process whereby owners of tokens stake their tokens on some outcome as a response to a query.
\medskip

\noindent \textbf{Definition} (\textit{Fork}).
The algorithm $\mathpzc{F}$, referred to as the \emph{fork}, accepts as input a tuple $(\mathtt{E}, \Omega, T)$ -- where $\mathtt{E}$ is an event, $\Omega$ is an outcome space of $\mathtt{E}$, and $T$ is a finite set of tokens -- and returns an $\Omega$-partition of $T$.

At a high level, $\mathpzc{F}$ works as follows. Each owner of tokens in $T$ is queried to ask which outcome in $\Omega$ is \texttt{True}. The owner's tokens are assigned to the cell that corresponds to the outcome they reported. If the owner does not respond (or if their response is not in $\Omega$), then their token are put in cell $C_{\texttt{Abstain}}$. Once all tokens in $T$ have been assigned to a cell, $\mathpzc{F}$ returns the collection of cells, which is an $\Omega$-partition of $T$. In pseudocode:

\begin{lstlisting}[language=json,escapeinside={*}{*}]
def *$\mathpzc{F}({\mathtt{E}},\Omega, T)$:*
  
  //begin with all cells empty
  for each *$\omega \in \Omega \cup \left\{ \texttt{Abstain} \right\}$*:
    *$C_{\omega} \leftarrow \emptyset$*
  endforeach

  //query all token owners for reports
  for each owner *j* of tokens in *$T$*:
    *$( \omega, R ) \leftarrow Report(j,\mathtt{E}, \Omega, T)$*
    //put the owner's tokens in the cell corresponding to the reported outcome
    *$C_{\omega} \leftarrow C_{\omega} \cup R$*
  endforeach
  
  //return the *$\Omega$-partition* of *$T$*
  return *$\left\{ C_{\texttt{Abstain}}, C_{\omega_{1}}, \ldots, C_{\omega_{|\Omega|}} \right\}$*  

enddef
\end{lstlisting}
\medskip

For our purposes, all calls to $\mathpzc{F}$ are assumed to be public, as are the resulting outputs. (In practice, $\mathpzc{F}$ would be implemented as a smart contract on a blockchain, and its entire history of calls and responses would be public.) The \texttt{for} loop in which token owners are queried may be run in parallel so that all token owners are queried simultaneously.

Finally, we define two simple subroutines: $Pay$ and $PluralityWinner$.
\medskip

\noindent \textbf{Definition} (\textit{Pay}).
The subroutine $Pay$ accepts as input a tuple $(T,\phi)$, where $T$ is a finite set of tokens and $\phi$ is some amount of currency. Each owner of tokens in $T$ is given a pro rata share of $\phi$. In particular, if $R \subseteq T$ is the set of tokens owned by some agent, then that agent will be paid a total of $\frac{\phi|R|}{|T|}$.
\medskip

\noindent \textbf{Definition} (\textit{PluralityWinner}).
The subroutine $PluralityWinner$ accepts as input an $\Omega$-partition and returns an outcome in $\Omega$ whose corresponding cell in the $\Omega$-partition is of the maximum size. Any ties are broken uniformly at random. (How ties are broken is unimportant for our purposes. Our results remain unchanged so long as the winner is chosen from among those outcomes whose corresponding cells have the maximum size.) In pseudocode:

\begin{lstlisting}[language=json,escapeinside={*}{*}]
def *$PluralityWinner(\{ C_{\texttt{Abstain}}, C_{\omega_{1}}, \ldots, C_{\omega_{|\Omega|}} \})$*
  //get outcomes in *$\Omega$* with largest corresponding cells
  *$ X \leftarrow \{ \omega \mid \omega \in \Omega \wedge \forall \gamma \in \Omega : |C_{\omega}| \ge |C_{\gamma}| \} $*
  
  //break any ties uniformly at random
  *$\hat{\omega} \xleftarrow{\$} X$*
  
  return *$\hat{\omega}$*
  
enddef
\end{lstlisting}
\medskip

Colloquially, $PluralityWinner$ simply interprets an $\Omega$-partition as the outcome of a plurality vote and returns the winner. Note that $PluralityWinner$ never returns \texttt{Abstain}, because $X \subset \Omega$ and $\texttt{Abstain} \notin \Omega$.

\section{Assumptions}

We model all agents as being rational. In particular, all agents come equipped with a von Neumann-Morgenstern utility function and always prefer actions that maximize their expected utility. For simplicity, we model agents as being risk neutral and having utility functions that are quasilinear in money. We assume agents are not budget constrained.

We further model all agents as being able to engage in costless communication with one another before making decisions. For the majority of the paper we will assume that the set of queriers and the set of reporters are disjoint, and we model the players as \emph{not} being able to make binding agreements with one another -- and so our first analyses will be done in the non-cooperative setting.

In section \ref{section:cooperative_model} we consider the effects of costless binding agreements and transferable utility by analyzing our approach in the cooperative model. Note that this also covers the case where the set of queriers and the set of reporters are not necessarily disjoint.\footnote{This is because the utility of a player that is both a querier and a reporter is the sum of the utilities of their ``queirer role" and their ``reporter role". So for the purposes of this analysis, the utility of a single ``querier-reporter" is indistinguishable from that of a coalition containing a querier and some positive number of reporters.}

Throughout the paper we assume that the center of each mechanism is a smart contract that has no \emph{a posteriori} knowledge of the world, and that the platform on which the smart contracts are executed is censorship resistant.

\section{The Simple Oracle, $\mathpzc{A}_{0}$}

\subsection{Motivation}

We construct a simple decentralized oracle that treats the fork $\mathpzc{F}$ as a plurality vote and returns the winning outcome. The key is to wrap $\mathpzc{F}$ with a mechanism that encourages token owners to report \texttt{True} when they are queried during $\mathpzc{F}$. When a certain ``economic soundness'' condition is met (see section \ref{section:economic_soundness_condition}) we can expect the winning outcome of the fork to be the \texttt{True} outcome of the event. In its most basic form, the algorithm consists of:

\begin{itemize}
\item Beginning with a reporting pool of tokens of equal value
\item Paying a reporting fee to the owners of tokens in the reporting pool before calling $\mathpzc{F}$ and
\item Permanently removing from the reporting pool any tokens that were not used to report \texttt{True} during the previous query 
\end{itemize}

Tokens removed from the reporting pool no longer earn their owners a reporting fee, so they are expected to have strictly lower value than tokens that remain in the pool. The price difference serves as an incentive for agents to report \texttt{True}.

\subsection{Construction}

The simple oracle, denoted $\mathpzc{A}_{0}$, works as follows. We create an initial finite set, $T_{genesis}$, of tokens that have no intended value outside of their use in this context. This set of tokens serves as the initial reporting pool. When querying the oracle, the caller pays the oracle a fee, denoted $\phi$. The oracle distributes this fee to owners of tokens in the reporting pool.

The oracle passes the query to the fork, which asks the owners of tokens in the reporting pool to report the \texttt{True} outcome of the event. The response from the fork is interpreted as the outcome of a plurality vote, and the outcome with the most votes is returned by the oracle.

It is those that query the oracle -- assisted by smart contracts -- that execute the algorithm $\mathpzc{A}_{0}$. In particular, it is those that query the oracle that determine which tokens were used to report truthfully during the previous call, and thus which token owners will be paid during the next call. While token owners decide the outcome that the oracle returns, it is those that query the oracle in the future that determine whether the previous response was true.

After the oracle returns an outcome, all tokens in the reporting pool that were not used to tell the truth are removed from the reporting pool for the next round. In pseudocode:

\begin{lstlisting}[language=json,escapeinside={*}{*}]
//initial state
*$C_{\texttt{True}}^{0} \leftarrow T_{genesis}$*
*$i \leftarrow 0$*
  
def *$\mathpzc{A}_{0}(\mathtt{E},\Omega,\phi)$*:
  //increment query counter
  *$i \leftarrow i+1$*
  
  //update the reporting pool
  //only truth-tellers from previous query remain in pool
  *$T_{i} \leftarrow C_{\texttt{True}}^{i-1}$*
  
  //pay owners of tokens in *$T_{i}$*
  *$Pay(T_{i},\phi)$*
  
  //call *$\mathpzc{F}$* with inputs *$({\mathtt{E}},\Omega,T_{i})$*
  *${\{ C_{\texttt{Abstain}}^{i}, C_{\omega_{1}}^{i}, \ldots, C_{\omega_{|\Omega|}}^{i} \} \leftarrow \mathpzc{F}(\mathtt{E},\Omega,T_{i})}$*
  
  //select winning outcome
  *$\hat{\omega}_{i} \leftarrow PluralityWinner(\{ C_{\texttt{Abstain}}^{i}, C_{\omega_{1}}^{i}, \ldots, C_{\omega_{|\Omega|}}^{i} \})$*
  
  //return winning outcome
  return *$\hat{\omega}_{i}$*
  
enddef
\end{lstlisting}

\section{Analysis of $\mathpzc{A}_{0}$}\label{section:analysis_of_a_0}

\subsection{Introduction}\label{section:simple_oracle_anaysis_intro}
Our goal is to design an incentive compatible, individually rational mechanism that implements a decision function that outputs the \texttt{True} outcome for a valid oracle query. We will show that when a certain ``economic soundness condition'' is satisfied (see Section \ref{section:economic_soundness_condition}) the simple oracle $\mathpzc{A}_{0}$ is such a mechanism.

The execution of the oracle $\mathpzc{A}_{0}$ is modeled as a repeated game with two classes of players: reporters and a querier. Each stage game of the repeated game is associated with an oracle query, and is modeled as a sequential game that operates as follows:

\begin{enumerate}
\item Each reporter chooses an outcome in response to the oracle query associated with the current stage game
\item The querier chooses how to update the reporting pool
\end{enumerate}

We will show that $\mathpzc{A}_{0}$ is incentive compatible by showing that there exists a Pareto efficient, subgame-perfect Nash equilibrium in the stage game which results in $\mathpzc{A}_{0}$ returning the \texttt{True} outcome for the oracle query. In particular, we will show that our desired player behavior -- where every reporter always reports the \texttt{True} outcome and the querier always removes from the reporting pool all and only those tokens used to lie -- is in equilibrium in the stage game. We will then show that $\mathpzc{A}_{0}$ is individually rational by demonstrating that the payouts for all players at this equilibrium are positive and strictly greater than their minmax payouts.

\subsection{The Economic Soundness Condition}\label{section:economic_soundness_condition}

Let $I_{i,j}$ denote the benefit to reporter $j$ from the oracle returning a false outcome in response to the $i$th oracle query. It is important to note that $I_{i,j}$ is intended to capture \textit{all} benefit to reporter $j$ -- including any ``extraneous" benefit --  from the oracle returning a false outcome in response to the $i$th oracle query. These benefits may be paid out in any currency. (We do not assume, for instance, that these benefits are paid out with tokens in $T$.) For example, if the oracle query is being used to determine payouts on a prediction market for a national election, and reporter $j$ has placed a large bet on the losing candidate, the value $I_{i,j}$ would include the value of reporter $j$'s bet. Similarly, any losing secondary bets (\textit{e.g.}, derivatives or other side-bets on the outcome, which may be denominated in entirely different currencies) placed by reporter $j$ are also included in the value $I_{i,j}$. 

Let $I_i=\sum\limits_{j}^{} I_{i,j}$ denote the total benefit received by all reporters from the oracle returning a false outcome in response to the $i$th oracle query. This represents the maximum total collective benefit (over all reporters) that could be gained from the oracle returning a false outcome to the $i$th query.

Let $p_i$ denote the market price of a token in the $i$th reporting pool $T_{i}$, and let $p_{i}^{\prime}$ denote the market price of a token in $T_{genesis} \setminus T_{i}$. (That is, $p_{i}^{\prime}$ denotes the market price of a token that has been removed from the reporting pool for lying.)

\medskip
\noindent \textbf{Definition} (\textit{Economic soundness condition}).
We say the \textit{economic soundness condition} is satisfied for the $i$th oracle query if $I_{i} < \frac{1}{2}(p_{i+1}-p_{i+1}^{\prime})|T_{i}|$.
\medskip 

When $i$ can be inferred from context, we may omit the subscripts and simply say that the economic soundness condition is satisfied when $I < \frac{1}{2}(p-p^{\prime})|T|$. If one assumes that tokens which have been removed from the reporting pool have zero value (that is, in the case where $\forall i, p^{\prime}_{i} = 0$), the economic soundness condition can be expressed as $I < \frac{1}{2}p|T|$, and can be interpreted as saying that the total collective benefit of causing the oracle to lie is less than half of the market cap of the reporting pool.

The motivation behind this definition is as follows. In order for the oracle $\mathpzc{A}_{0}$ to return a false outcome in response to the $i$th oracle query, at least half of all tokens in the $i$th reporting pool -- that is, at least $\frac{1}{2} |T_{i}|$ tokens -- must be used to lie or abstain (otherwise, the \texttt{True} outcome would necessarily receive the most votes and thus become the winner). Each token used to lie or abstain loses $p_{i+1}-p^{\prime}_{i+1}$ in value. Thus the minimum total cost of causing $\mathpzc{A}_{0}$ to return a false outcome in response to the $i$th oracle query is given by $\frac{1}{2}(p_{i+1}-p_{i+1}^{\prime})|T_{i}|$.

Colloquially, then, the economic soundness condition is satisfied exactly when the total cost of forcing the oracle to lie exceeds the total collective benefit -- including all ``extraneous" benefit -- from doing so. As we will show, this condition is necessary and sufficient for $\mathpzc{A}_{0}$ to be an incentive compatible and individually rational implementation of our desired truth-telling decision function.

It is not surprising that our most important results are predicated on the economic soundness condition being satisfied. All incentive compatible oracles -- even centralized ones -- necessarily have an analogous soundness condition: if the cost of causing the oracle to lie is less than one could steal by doing so, it would be irrational not to cause the oracle to lie. Unsurprising as this may be, it is important not to take economic soundness for granted. We are not, in general, guaranteed to have the economic soundness condition be satisfied, even if the oracle has been reporting \texttt{True} outcomes since its genesis. We investigate the conditions under which we may expect the economic soundness condition to be satisfied in practice in section \ref{section:simple_oracle_tenability_of_economic_soundness}.

\begin{figure*}
\centering
\includegraphics[width=3.75in]{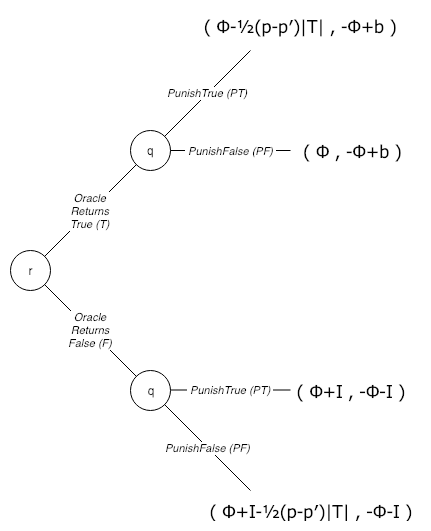}
\caption{The stage game shown in extensive form, as described in the proof of Theorem \ref{th:if_sound_and_true_then_punish_false}, with the set of reporters modeled as a single player ($r$) who may unilaterally decide the outcome of the oracle at the minimum possible cost. The querier ($q$) then decides which set of coins to remove from the reporting pool. Observe that, at each move, the querier is indifferent among her available moves, and therefore it is trivially the case that \textit{every} response by the querier is a best response and \textit{every} Nash equilibrium is a subgame-perfect Nash equilibrium. This game is shown again in normal form in Figure \ref{fig:querier_decision_stage_game_normal_form}.}
\label{fig:querier_decision_stage_game_extensive_form}
\end{figure*}

\subsection{Incentive Compatibility}\label{section:simple_oracle_incentive_compatibility}

In this section we show that, when the economic soundness condition is satisfied, the simple oracle $\mathpzc{A}_{0}$ is an incentive compatible implementation of our desired truth-telling decision function. In other words, when the economic soundness condition is satisfied and all players are behaving the way we want them to, no player can do better for themselves by unilaterally deviating from that behavior. We do this in the standard way, by showing that the resulting game contains an equilibrium strategy profile that results in the oracle deciding the \texttt{True} outcome.

The strategy spaces in the stage game are modeled as follows. Each stage game is associated with an oracle query which comes with an event $\mathtt{E}$ and an outcome space $\Omega$. The strategy set for each reporter is $\{True,False\}$, modeling the choice reporters make during a fork. The strategy $True$ represents the reporter choosing to report the $\texttt{True}$ outcome during the fork, while the strategy $False$ represents the reporter abstaining or reporting a false outcome during the fork. Afterwards, the querier chooses how to update the reporting pool by choosing from $\{PunishFalse, PunishTrue\}$, where $PunishFalse$ represents the querier removing from the reporting pool any tokens used to abstain or lie during the fork, and $PunishTrue$ represents the querier removing from the reporting pool any tokens used to report the \texttt{True} outcome during the fork.

\medskip
\noindent \textbf{Definition} (\textit{Honest play}).
Let \textit{honest play} refer to the strategy profile in which every reporter always chooses to report \texttt{True} and the querier always chooses the move $PunishFalse$.
\medskip

The following two theorems establish that honest play is in equilibrium in the stage game.

\begin{figure*}
\centering
\includegraphics[width=5in]{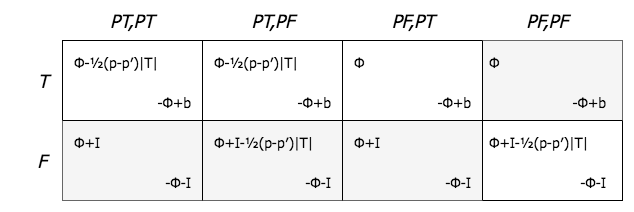}
\caption{The stage game from Figure \ref{fig:querier_decision_stage_game_extensive_form} shown in normal form. The set of reporters (modeled here as a single player as described in the proof of Theorem \ref{th:if_sound_and_true_then_punish_false}) is the row player and the querier is the column player. The pure strategy Nash equilibria -- when the economic soundness condition is satisfied -- are highlighted in gray.}
\label{fig:querier_decision_stage_game_normal_form}
\end{figure*}

\begin{theorem}\label{th:if_sound_and_true_then_punish_false}
If the economic soundness condition is satisfied then always choosing the move $PunishFalse$ is a best response by the querier to any strategy profile chosen by the reporters that results in the oracle returning $\texttt{True}$.
\end{theorem}

\begin{proof}
See appendix.
\end{proof}

\begin{theorem}\label{th:if_sound_and_punish_false_then_true}If the economic soundness condition is satisfied and the querier always chooses the move $PunishFalse$, then reporting the \texttt{True} outcome is always the best response by every individual reporter.
\end{theorem}

\begin{proof}
See appendix.
\end{proof}

As an immediate result of Theorems \ref{th:if_sound_and_true_then_punish_false} and \ref{th:if_sound_and_punish_false_then_true} we can see that honest play is in equilibrium in the stage game. Moreover, the payouts (in the stage game) to all players during honest play are strictly greater than their minmax payouts: the minmax payout is $\phi-r_{j}|T|$ for the individual reporter $j$, and $-\phi - I$ for the querier. So, by Friedman's folk theorem~\cite{Friedman_1971}, honest play is also in equilibrium in the repeated game. (In fact, the folk theorems tell us that honest-play is in equilibrium in the infinitely repeated game without discounting, the infinitely repeated game \textit{with} discounting, and the finitely repeated game without discounting. This is nice, as it means our incentive compatibility result is robust against our choice of repeated-game model.) In other words, no individual player has any incentive to deviate from honest-play, and so our mechanism is incentive compatible.

\begin{theorem}\label{th:if_sound_then_ic}If the economic soundness condition is satisfied then $\mathpzc{A}_{0}$ is an incentive compatible mechanism that decides the \texttt{True} outcomes of oracle queries.
\end{theorem}

\begin{proof}
The result follows immediately from Theorems \ref{th:if_sound_and_true_then_punish_false} and \ref{th:if_sound_and_punish_false_then_true}.
\end{proof}

\subsection{Individual Rationality}\label{section:simple_oracle_individual_rationality}

We have shown that when the economic soundness condition is satisfied, honest-play is in equilibrium for all players engaging with the oracle. Of course, in the real world, we cannot \textit{force} agents to engage with our mechanism. If they are to engage, they must do so willingly. Since players have the choice of not participating, we want our mechanism to satisfy an individual rationality constraint. That is, we want it to be the case that the outcomes for honest-play are better (or at least not worse) for all players than they would achieve by not playing at all. To show that our mechanism is individually rational, we must show that the equilibrium at honest-play results in a non-negative payout for all players.

As in the proof of Theorem \ref{th:if_sound_and_true_then_punish_false}, let $b$ denote the benefit the querier receives when the oracle returns the \texttt{True} outcome, and recall that $\phi$ denotes the fee paid by the querier to the reporters.

\begin{theorem}\label{th:if_sound_and_small_fee_then_ir}If the economic soundness condition is satisfied and $b > \phi$ then the simple oracle $\mathpzc{A}_{0}$ is individually rational.
\end{theorem}

\begin{proof}
Suppose the economic soundness condition is satisfied. Then, since honest play is in equilibrium, it will suffice to show that the outcomes for each player during honest play are non-negative.
From Figure \ref{fig:querier_decision_stage_game_normal_form} we can observe that honest play results in individual reporters being paid a pro rata share of $\phi$, which is always non-negative. Also from Figure \ref{fig:querier_decision_stage_game_normal_form} we can see that the querier receives a payout of $-\phi + b$, which is non-negative when $b > \phi$. Thus, when the economic soundness condition is satisfied and $b > \phi$, the payout to all players is non-negative during honest play.
\end{proof}

\subsection{Tenability of Economic Soundness}\label{section:simple_oracle_tenability_of_economic_soundness}
Since all of the results above depend upon the economic soundness condition being satisfied, we will now examine whether it is reasonable to expect the economic soundness condition to be satisfied in practice. Naturally, speculation on the future value of a token can result in arbitrarily high token prices, so it is certainly always \textit{possible} for the economic soundness condition to be satisfied. However, we are interested in whether the reporting fees \textit{alone} can justify a high enough token price for the economic soundness condition to be satisfied.

In particular, we want to know whether the reporting fee can, simultaneously, be small enough that the querier is willing to pay it and large enough to make the market price of tokens (and therefore the market cap of the reporting pool) high enough to satisfy the economic soundness condition. As we will show, this depends very much on the market's current appetite for risk, how the reporting fee is chosen (as a function of $I$ and time), and how high a fee the queriers are willing to bear.

To aid the discussion, consider a betting platform that uses an instance of the oracle $\mathpzc{A}_{0}$ to report the outcomes of national elections. At any given time, the total value of all open bets on the platform is referred to as the \textit{open interest}, which is the maximum benefit that the set of reporters could gain by making the oracle lie: malicious reporters can steal open interest by betting on low-likelihood outcomes and then forcing the oracle to resolve to the outcome on which they bet. So, in this case, $I$ is the open interest.

To consider the simplest case, suppose that the querier (which, in this case, is the betting platform) always chooses $PunishFalse$, and that any tokens used to lie have zero value (that is, $\forall i, p^{\prime}_{i} =0$). Thus the minimum cost of causing the oracle to lie is $\frac{1}{2}p|T|$ and the economic soundness condition is satisfied if and only if $I < \frac{1}{2}p|T|$.

Now suppose the betting platform charges the bettors a fee that is some percentage $x$ of their bet size. Every bet incurs a fee, and all fees collected by the platform will be pooled together and used as the reporting fee when it queries the oracle to decide the outcome of an election. So $\phi = x I$ for every oracle query. For simplicity, assume that every election results in the same volume of bets, and so $I$ remains constant over multiple oracle queries.

Finally, let the market's expected current yield for tokens in $T$ be $Y$. In other words, we expect the market to behave in such a way that $Y=\frac{A}{p|T|}$, where $A$ is the sum of all reporting fees collected over one year. Therefore, if the betting platform makes $n$ oracle queries in one year, then $A=nxI$, and market behavior will result in a token price of $p=\frac{nxI}{Y|T|}$. Using this value for $p$ in the definition of the economic soundness condition, we can see that we may expect the economic soundness condition to be satisfied when $x>\frac{2Y}{n}$.

As we can see from this simple example alone, the tenability of the economic soundness condition is dependent on factors outside of the implementer's control. While the creators of the betting platform may be able to exert control over the fees they charge ($x$) and the number of times that they query the oracle in a given year ($n$), they cannot control the market's expectation of current yield ($Y$) or whether not their users are \textit{willing to pay} a high enough fee ($x$) to satisfy the condition $x>\frac{2Y}{n}$.

\medskip
\textbf{Example 1:}
Suppose the market expects a current yield of $30\%$ for holding tokens in $T$, and the betting platform is making one oracle query per month (so that $Y=0.3$ and $n=12$). Then if the bettors are unwilling to pay a fee of $5\%$ on their bets, we should not expect the economic soundness condition to be satisfied.

\medskip
\textbf{Example 2:}
Suppose the market expects a current yield of $25\%$ and the betting platform is making one oracle query per week (so that $Y=0.25$ and $n=52$). Then if the bettors are willing to pay a fee of $1\%$ on their bets, we should expect the economic soundness condition to be satisfied.
\medskip

In conclusion, we are not guaranteed to have the economic soundness condition be satisfied \textit{in general}. Under some conditions it is satisfied quite easily, and under other conditions it is not. Moreover, the tenability of the economic soundness condition depends on some factors outside of the oracle implementer's control, such as the market's appetite for current yield and user tolerance to the minimum required fees.

In plain terms, this is not a simple, drop-in ``oracle solution" that is certain to work for any project that needs an oracle. Projects that are considering implementing the oracles presented in this paper should give special consideration to how they structure their fees, the market's appetite for current yield, and whether their users will be tolerant of the minimum required fees.

\subsection{Weaknesses}\label{section:simple_oracle_weaknesses}

This mechanism achieves only a weak version incentive compatibility (Bayesian-Nash incentive compatibility). While honest play is in equilibrium, and while its resulting equilibrium is Pareto efficient and subgame-perfect, the strategies played during honest play are not strictly dominant. Indeed, the game that arises from $\mathpzc{A}_{0}$ has no strictly dominant strategies for any of the players at all. Reporting truthfully is a best response for individual reporters only if the querier always chooses $PunishFalse$. However, always choosing $PunishFalse$ is only a \textit{weakly} dominant strategy for the querier.

The truthfulness of $\mathpzc{A}_{0}$ hinges upon the querier choosing one particular weakly dominant strategy from among several available to her. As we have shown above, it is rational for her to always choose $PunishFalse$ during each stage game.  However, it is also rational for her to choose one of her other weakly dominant strategies (in the stage game setting). After all, as shown in Figure \ref{fig:querier_decision_stage_game_extensive_form}, the querier is indifferent between her available moves during the stage game. It is only when the querier considers the larger repeated game setting that always choosing $PunishFalse$ becomes more appealing than the alternatives.

Of course, queriers are unlikely to engage with the mechanism at all unless they intend to engage in honest play. This is because the mechanism is not individually rational outside of honest play, and so the querier would do better for herself by not engaging with the mechanism at all than to engage and play dishonestly. Nevertheless, the situation would be greatly improved if the querier's preference for honest play were strict.

An ideal mechanism would be \textit{dominant strategy incentive compatible} (DSIC), so that players could choose to play honestly without having to give any consideration to the behavior of other players. It is an open question whether there exists such a mechanism that can be implemented in a setting where the players have common knowledge of the truth but the center of the mechanism (\textit{i.e.}, the smart contract) does not.

A less ambitious goal would be to design a variation of $\mathpzc{A}_{0}$ for which always choosing $PunishFalse$ were a strictly dominant strategy for the querier in the stage game, even if there were still no dominant strategy available for individual reporters. While reporters would still have to reason about the future behavior of the querier before deciding whether to report truthfully, the resulting stage game would have just one equilibrium: honest play. This would be a clear improvement upon $\mathpzc{A}_{0}$, as $\mathpzc{A}_{0}$ results in multiple equilibria in the stage game -- only one of which results in the oracle returning \texttt{True}. As with the decentralized DSIC mechanism, it is an open question whether there exists a mechanism with this property than can be implemented in the setting where the center of the mechanism (\textit{i.e.}, the smart contract) does not posses common knowledge of truth.

Finally, this approach necessitates that the economic soundness condition be satisfied. As we saw in section \ref{section:simple_oracle_tenability_of_economic_soundness}, this condition is dependent upon things outside of the oracle implementer's control. It is impossible, for example, for the oracle implementer to prevent third-party derivatives being resolved by the oracle's outputs. These derivative bets can be made without the secondary bettors paying any fee to the reporters. As a result, third-party derivatives increase the value of $I$ but may not increase the market cap of reporting tokens, and thereby jeopardize the incentive compatibility of the oracle.  This is known as the ``parasite problem" and we conjecture that it is unsolvable for all public oracles, both centralized and decentralized.

\section{Scaling Strategy}
Real life execution of the simple oracle $\mathpzc{A}_{0}$ requires that participants agree on the state of the reporting pool $T$. The correct state of the reporting pool cannot be verified by a smart contract alone -- in particular, determining which outcomes were \texttt{True} for each previous oracle query cannot be done by a smart contract. It requires \textit{a posteriori} knowledge of what was common knowledge during each of the previous calls of $\mathpzc{F}$ by $\mathpzc{A}_{0}$. To verify the correctness of the current reporting pool, a new user must examine the output of \textit{every} previous call of $\mathpzc{F}$ by the oracle. For each $\Omega$-partition returned by $\mathpzc{F}$, the new user must decide which set of tokens corresponds to $C_{\texttt{True}}$. Only then can they determine the correct state of the current reporting pool. In brief, on-boarding a new user requires the new user to manually determine the \texttt{True} outcome for every previous oracle query. This does not scale.
 
In the following sections, we define and analyze oracles that do not need to call $\mathpzc{F}$ on every oracle query. Instead, we allow queriers to submit a proposed outcome along with their oracle query, and we give the reporters an opportunity to dispute the proposed outcome if they think it is false. If the proposed outcome is not disputed, then the oracle returns the proposed outcome without having to call the fork. (Our assumption that the smart contract platform is censorship resistant is critical here. An attacker that can censor dispute transactions can cause these new oracles to return false outcomes by preventing reporters from disputing false proposed outcomes.) If the proposed outcome \textit{is} disputed, then the oracle uses the fork to determine which outcome to return, just like $\mathpzc{A}_{0}$.

This creates a subgame for each oracle query, where the querier chooses whether to submit their query with the \texttt{True} outcome or a false outcome as the proposed outcome, then reporters decide whether or not to dispute the proposed outcome. We leverage the credible threat of a fork along with some bonds to make honest play incentive compatible in the subgame. The result is that, when the economic soundness condition is satisfied, the oracle is expected to return the \texttt{True} outcome of a query without having to call $\mathpzc{F}$, so the reporting pool does not get updated after every oracle query. On-boarding new users does not require them to manually verify the results of every previous oracle query, but only those which required a fork.

\section{An Oracle with a Single Dispute Round, $\mathpzc{A}_{1}$}
\subsection{Construction}
For this new oracle we introduce a \textit{dispute round} which leverages bonds and the credible threat of the fork. When the oracle is queried, the query must be accompanied by a \textit{tentative outcome}, $\hat{\omega} \in \Omega$, and some initial \textit{stake} on that outcome. Then begins a period of time, referred to as a \textit{dispute round}, during which token owners have the opportunity to dispute the tentative outcome in favor of some other outcome in $\Omega$ by adding some specific amount of stake to their chosen outcome.

If no dispute takes place during the dispute round, then the oracle returns $\hat{\omega}$ and the initial stake is returned to its original owner. Otherwise the oracle calls $\mathpzc{F}$ to determine the winning outcome, just as before. Any stake -- Whether it was the initial stake that came with the oracle query or stake placed during the dispute round -- that was placed on a losing outcome is transferred to those who staked on the winning outcome. In this way, token owners are incentivized to dispute any tentative outcomes that would not win a fork in favor of outcomes that would win a fork.

We will show that, at equilibrium, the oracle returns the \texttt{True} outcome without the reporting pool having to be updated.

Later in this section we will construct an oracle that uses a single dispute round and discuss its strengths and limitations. In the following section, we will define an oracle that uses multiple dispute rounds to address those limitations.

First, we define a few algorithms that will be used as subroutines in the following oracles.

\medskip
\noindent \textbf{Definition} (\textit{DisputeRound}).
The algorithm $DisputeRound$ accepts a tuple $(\mathtt{E},\Omega,\hat{\omega},\textbf{D},s)$, where $\mathtt{E}$ is an event, $\Omega$ is an outcomes space of $\mathtt{E}$, $\hat{\omega} \in \Omega$ is a tentative outcome, $\textbf{D}=(D_{\omega})_{\omega \in \Omega \cup \left\{\texttt{Abstain} \right\} }$ is an $\Omega$-partition of the set of tokens that have been staked on some outcome $\omega \in \Omega$ during the present oracle query, an $s$ is the amount of dispute stake required to dispute the tentative outcome. We let $d=|D_{\hat{\omega}}|$, and we require that $D_{\hat{\omega}}$ not be empty. That is, we require that the dispute round begins with some positive amount, $d > 0$ of stake on the tentative outcome.

Let $T$ be the set of tokens in the reporting pool. All owners of tokens in $T \setminus \bigcup \textbf{D}$ (that is, tokens that are in the reporting pool but have not already been used to stake on some outcome during the current oracle query\footnote{Here we are using the notation $\bigcup\limits \textbf{D}$ to denote the union of sets in \textbf{D}: \[\bigcup\limits \textbf{D} = \bigcup\limits_{X \in \textbf{D}} X = D_{\texttt{Abstain}} \cup D_{\omega_{1}} \cup \ldots \cup D_{\omega_{|\Omega|}}\].}) have the opportunity -- but not the obligation -- to dispute the tentative outcome in favor of some other outcome in $\Omega$.

A dispute consists of staking $s$ reporting tokens (referred to as \textit{dispute stake} in this context) on some outcome other than the market's current tentative outcome. In this paper, for simplicity, we say that disputes require double the stake on the current tentative outcome. Our results remain unchanged if disputes are be made to be $\alpha$ times the stake on the current tentative outcome, so long as $\alpha > 1$.

Dispute rounds have a fixed maximum time limit. If a dispute occurs in favor of outcome $\omega$ before the time limit, then $DisputeRound$ updates $D_{\omega}$ to include the new dispute stake and returns $(\omega,\textbf{D},\texttt{TRUE})$. Otherwise, $DisputeRound$ does not modify any cell of the $\Omega$-partition $\textbf{D}$ and returns $(\hat{\omega},\textbf{D},\texttt{FALSE})$.

At most one token holder may actually dispute a tentative outcome during any given dispute round. That is, while all token holders have the opportunity to dispute a tentative outcome, at most one token holder can actually perform the dispute.
\medskip

\noindent \textbf{Definition} (\textit{Distribute}).
The algorithm $Distribute$ accepts as input a tuple $(\textbf{D},\hat{\omega})$ where $\textbf{D}$ is an $\Omega$-partition (of dispute stake) and $\hat{\omega} \in \Omega \cup \{\texttt{Abstain}\}$. The algorithm pays out all the tokens in $\bigcup{\textbf{D}}$ to those reporters who own tokens in $D_{\hat{\omega}}$, in proportion to the number of tokens they own in $D_{\hat{\omega}}$. That is, if a token holder owns $X$ tokens in $D_{\hat{\omega}}$ then $Distribute$ will pay that token holder $\frac{X}{|D_{\hat{\omega}}|} |\bigcup{\textbf{D}}|$ tokens.
\medskip

Colloquially, $Distribute$ simply distributes the tokens in $\textbf{D}$, pro rata, to the reporters who staked on outcome $\hat{\omega}$.
\medskip

\noindent \textbf{Definition} (\textit{ChoiceByFork}).
The use of the fork as a fallback for deciding the winning outcome is expressed in the subroutine $ChoiceByFork$, described in pseudocode here:

\begin{lstlisting}[language=json,escapeinside={*}{*}]
def *$ChoiceByFork(\mathtt{E},\Omega,T, \textbf{D})$*:
  //call the fork but let only non-dispute stake participate
  *${\{ C_{\texttt{Abstain}}, C_{\omega_{1}}, \ldots, C_{\omega_{|\Omega|}} \} \leftarrow \mathpzc{F}(\mathtt{E},\Omega,T\setminus \bigcup \textbf{D})}$*
      
  //select winning outcome (remembering to include the dispute stake)
  *$\hat{\omega} \leftarrow PluralityWinner(\{ C_{\texttt{Abstain}} \cup D_{\texttt{Abstain}}, C_{\omega_{1}} \cup D_{\omega_{1}}, \ldots, C_{\omega_{|\Omega|}} \cup D_{\omega_{|\Omega|}} \})$*
  
  //redistribute dispute stake
  *$Distribute(\textbf{D}, \hat{\omega})$*
  
  //put all dispute stake in the cell corresponding to the winning outcome
  *$C_{\hat{\omega}} \leftarrow \bigcup \textbf{D} \cup C_{\hat{\omega}}$*
  
  //return winning outcome and the *$\Omega$*-partition
  return *$(\hat{\omega}, \{ C_{\texttt{Abstain}}, C_{\omega_{1}}, \ldots, C_{\omega_{|\Omega|}} \})$*
enddef
\end{lstlisting}
\medskip

With these definitions in place we are ready to describe the oracle $\mathpzc{A}_{1}$ -- the oracle with a single dispute round. A query to $\mathpzc{A}_{1}$ must come along with a tentative outcome $\hat{\omega} \in \Omega$ and some initial stake of $d>0$ tokens, which we denote $\{t_{1}, \ldots, t_{d}\}$. When the query is received, all reporters have an opportunity to dispute the tentative outcome by staking $2d$ tokens on any outcome other than the tentative outcome. If no dispute occurs, the oracle returns the tentative outcome and returns the initial stake back to the querier. If some reporter \textit{does} dispute the tentative outcome, then the oracle calls the fork to determine the winner.

In the event that the tentative outcome is disputed, the oracle will use the fork to determine the winning outcome, just as we did with $\mathpzc{A}_{0}$. We require that any stake that was placed in favor of some outcome during a dispute round must be used to report the same outcome during the fork. That is, if a reporter stakes $10$ tokens on outcome $\omega_{1}$ during a dispute round and then the oracle calls the fork, then that player has no choice but to use those $10$ tokens to report outcome $\omega_{1}$ during the fork. In other words, if a player has staked on an outcome, then they remain committed to that outcome in the event that a fork is called.

In the event that the oracle calls the fork, the oracle returns whatever outcome wins the fork (just as with $\mathpzc{A}_{0}$). Additionally, the initial stake and the dispute stake are redistributed to whichever players staked on the outcome that won the fork. In pseudocode:

\begin{lstlisting}[language=json,escapeinside={*}{*}]
//initial state
*$C_{\texttt{True}}^{0} \leftarrow T_{genesis}$*
*$i \leftarrow 0$*
  
def *$\mathpzc{A}_{1}(\mathtt{E},\Omega,\phi,\hat{\omega},\{t_{1},\ldots,t_{d} \})$*:
  //increment query counter
  *$i \leftarrow i+1$*
  
  //update the reporting pool
  //only truth-tellers from previous query remain in pool
  *$T_{i} \leftarrow C_{\texttt{True}}^{i-1}$*
  
  //pay owners of tokens in *$T_{i}$*
  *$Pay(T_{i},\phi)$*
  
  //init *$\Omega$*-partition of dispute stake
  for *$\omega \in \Omega \cup \{\texttt{Abstain}\}$*:
    if *$\omega == \hat{\omega}$*:
      *$D_{\omega} \leftarrow \{t_{1},\ldots,t_{d} \}$*
    else:
      *$D_{\omega} \leftarrow \emptyset$*
    endif
    *$\textbf{D} \leftarrow \{ D_{\texttt{Abstain}}, D_{\omega_{1}}, \ldots, D_{\omega_{|\Omega|}} \}$*
  endfor
  
  //have a dispute round
  *$ (\hat{\omega}_{i},\textbf{D},\texttt{DISPUTED}) \leftarrow DisputeRound(\mathtt{E}, \Omega, \hat{\omega},\textbf{D}, 2d)$*
  
  //if there was no dispute
  if *$\texttt{DISPUTED} == \texttt{FALSE}$*:    
    //then give the initial stake back to the querier
    *$Distribute(\textbf{D},\hat{\omega}_{i})$*
    
    //no reporting tokens are removed from the reporting pool
    *$C_{\texttt{True}}^{i}\leftarrow T_{i}$*
        
    //return tentative outcome
    return *$\hat{\omega}_{i}$*
  
  //else if there was a dispute  
  elseif *$\texttt{DISPUTED} == \texttt{TRUE}$*:
    
    //resort to fork
    *$(\hat{\omega}_{i}, \{ C_{\texttt{Abstain}}^{i}, C_{\omega_{1}}^{i}, \ldots, C_{\omega_{|\Omega|}}^{i} \}) \leftarrow ChoiceByFork(\mathtt{E},\Omega, T_{i}, \textbf{D})$*
  
    //return outcome that won the fork
    return *$\hat{\omega}_{i}$*
 
 endif
enddef
\end{lstlisting}

\section{Analysis of $\mathpzc{A}_{1}$}
\subsection{Introduction}
We want to show that if the economic soundness condition is satisfied then, at equilibrium, the oracle $\mathpzc{A}_{1}$ returns the \texttt{True} outcome of an oracle query \textit{without having to invoke a fork}. In section \ref{section:simple_oracle_incentive_compatibility} we showed that honest play during a fork is in equilibrium when the economic soundness condition is satisfied. So for the purposes of this section it will suffice to show that, if honest play is expected during a fork, then there exists a unique subgame-perfect equilibrium in the subgame induced by the dispute round for which:

\begin{itemize}
\item The querier submits their query with the \texttt{True} outcome as the tentative outcome.
\item  When the tentative outcome is the \texttt{True} outcome, it is not disputed.
\item When a tentative outcome is false, it is disputed in favor of the \texttt{True} outcome.
\end{itemize}

We will show that, when honest play is expected during forks, there exists a unique subgame perfect equilibrium in the dispute round game that satisfies these three properties.

\subsection{Unique Equilibrium}
The dispute round game is modeled as a sequential game in which the querier moves first by choosing whether to submit their query with the \texttt{True} outcome as the tentative outcome, or some false outcome as the tentative outcome. The querier necessarily stakes $d$ tokens on this tentative outcome. Then, the reporters choose whether or not to dispute the tentative outcome. If a reporter chooses to dispute the tentative outcome, they must do so by staking $2d$ tokens in favor of some other outcome. We are interested in the case where honest play is expected during a fork, so any play that results in a fork pays out as if the \texttt{True} outcome were returned.

\begin{theorem}\label{thm:single-dispute-round-incentive-compatibility}
When honest play is expected during a fork, there exists a unique subgame perfect equilibrium in the single dispute round game. The unique subgame perfect equilibrium results in the following behavior:
\begin{itemize}
\item The querier submits their query with the \texttt{True} outcome as the tentative outcome
\item If the querier submits the \texttt{True} outcome as the tentative outcome, then no reporter disputes.
\item If the querier submits a false outcome as the tentative outcome, then there exists a reporter who disputes in favor of the \texttt{True} outcome.
\end{itemize}
\end{theorem}

\begin{proof}
The strategic decisions facing the querier and an arbitrary reporter during the single dispute round game are illustrated in Figure \ref{fig:single_dispute_round_game_extensive_form}. Reducing the chance moves to their expected values results in a simplified game tree shown in Figure \ref{fig:single_dispute_round_game_extensive_form_solved}, which is solved by backward induction to see the unique subgame perfect equilibrium (also shown in Figure \ref{fig:single_dispute_round_game_extensive_form_solved}), which results in our three desired behaviors.
\end{proof}

\begin{figure*}
\centering
\includegraphics[width=5in]{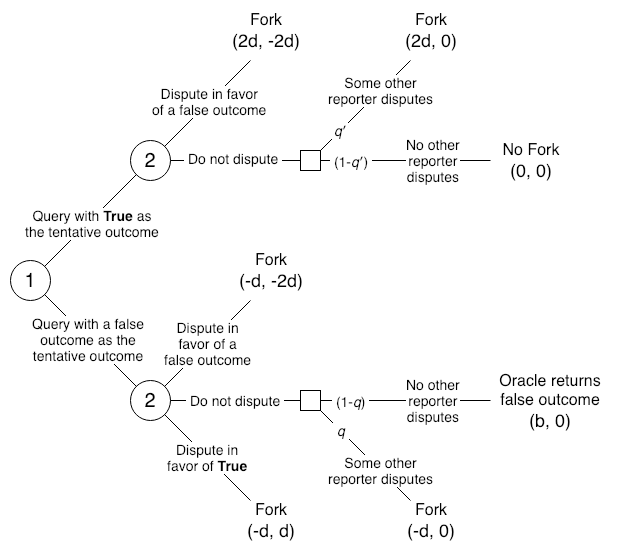}
\caption{The single dispute round game in extensive form. The value $\textbf{d}$ is the amount of stake on the initial tentative outcome. The value $\textbf{b}$ is the benefit to player 1 if the oracle returns a false outcome. Player 2 is an arbitrary single reporter, and the effects of the decisions made by the remaining reporters are modeled as chance moves.}
\label{fig:single_dispute_round_game_extensive_form}
\end{figure*}

\begin{figure*}
\centering
\includegraphics[width=3.5in]{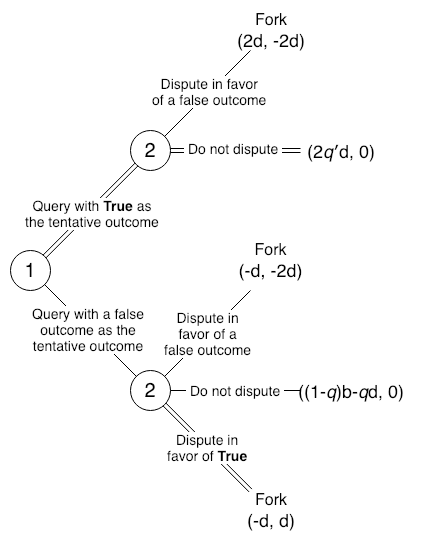}
\caption{The single dispute round game from Figure \ref{fig:single_dispute_round_game_extensive_form} with chance moves replaced with expected payouts. The unique subgame-perfect equilibrium is derived via backward induction and is shown with doubled edges.}
\label{fig:single_dispute_round_game_extensive_form_solved}
\end{figure*}

By Theorem \ref{thm:single-dispute-round-incentive-compatibility}, our desired player behavior is incentive compatible when honest play occurs during a fork. By Theorem \ref{th:if_sound_then_ic}, honest play during a fork is incentive compatible when the economic soundness condition is satisfied. It follows that, when the economic soundness condition is satisfied, the oracle $\mathpzc{A}_{1}$ is an incentive compatible implementation of our desired truth-telling function that does not need to call the fork on every oracle query.

Moreover, one can quickly verify that $\mathpzc{A}_{1}$ is individually rational by observing that the payouts for all players are non-negative when playing according to the unique subgame-perfect equilibrium shown in Figure \ref{fig:single_dispute_round_game_extensive_form_solved}.

\subsection{Weakness}
While $\mathpzc{A}_{1}$ does satisfy all of our stated design goals, it has the unfortunate property that an attacker can grief honest participants by intentionally causing many forks. While such behavior would be irrational in the context of our simple model, it may well be incentivized by extraneous circumstances when interacting with $\mathpzc{A}_{1}$ in the real world. (For example, a competing oracle platform may be financially motivated to damage the scalability properties of $\mathpzc{A}_{1}$ by intentionally causing many forks -- perhaps in the hopes of steering new users to their own platform.) We address this weakness in the following section.

\section{An Oracle with Multiple Dispute Rounds, $\mathpzc{A}_{2}$}
\subsection{Motivation}
Note that the minimum cost of causing a fork during a call of $\mathpzc{A}_{1}$ is exactly the same as the cost of querying $\mathpzc{A}_{1}$ with a false outcome. An attacker who desires to cause a fork (so as to negatively impact scalability) can simply query the oracle with a false tentative outcome. This outcome will be disputed, causing a fork, and costing the attacker $d$ tokens per query.

If $d$ is small, then it is cheap for honest users to query the oracle, but it is also cheap for an attacker to cause many forks. If $d$ is large, then causing many forks is expensive, but the capital needed to query the oracle in the first place may be prohibitive for honest users. Oracle implementers may want to keep the capital required to query the oracle small, while also increasing the cost of causing a fork.

There may be many ways to achieve this property.\footnote{A naive approach is to require that the querier submit their query with some small amount $d$ staked on the initial tentative outcome, but then require, say, $10000d$ stake in order to dispute the tentative outcome. While this may technically work in our simple, abstract model with no consideration of external investment opportunities, it is unlikely to work in practice. In particular, it is unlikely that a real-life honest participant would lock up capital to dispute a false tentative outcome for an ROI of just $0.01\%$. By contrast, our proposed approach guarantees disputers of false outcomes (in favor of \texttt{True} outcomes) an ROI of $40$ - $50\%$. We think this approach is likely to induce our desired player behavior in practice.} Here we will focus on just one: using multiple consecutive dispute rounds.

\subsection{Construction}
For this new oracle we introduce a \textit{dispute sequence}, which is simply a finite sequence of dispute rounds. As before, the query must be accompanied by a tentative outcome, $\hat{\omega}_{1} \in \Omega$, and some initial stake of $d$ tokens. Then the dispute sequence begins.

During the first dispute round in the dispute sequence, token holders have an opportunity to dispute the current tentative outcome by staking $2d$ tokens on any outcome \textit{other than} the current tentative outcome. If nobody disputes $\hat{\omega}_{1}$, then the $d$ tokens are returned to the querier, the dispute sequence ends, and the oracle outputs $\hat{\omega}_{1}$.

However, if $\hat{\omega}_{1}$ \textit{is} disputed, then a new dispute round begins, with the newly championed outcome, $\hat{\omega}_{2}$, as the tentative outcome for the new dispute round. Once again, all token holders have the opportunity to dispute the new tentative outcome, $\hat{\omega}_{2}$, in favor of any other outcome -- but this time they are required to stake $4d$ tokens on the newly championed outcome.

This continues either until some tentative outcome survives a dispute round without being disputed, or until the dispute stake posted by a disputer reaches some threshold $M$ (a constant chosen by the implementers of the oracle). If the former occurs, the oracle returns the tentative outcome without calling the fork. If the latter, then the oracle calls the fork, just as was done for $\mathpzc{A}_{1}$.

In general, the $n$th dispute round has tentative outcome $\hat{\omega}_{n}$, and token holders can dispute $\hat{\omega}_{n}$ in favor of any other outcome by staking $2^{n-1} d$ tokens on the newly championed outcome. If no such dispute occurs then the dispute sequence ends, the oracle will output $\hat{\omega}_{n}$, and all token holders who have staked on $\hat{\omega}_{n}$ (during any dispute round) will receive their stake back, in addition to receiving a pro rata share of all stake that was staked on outcomes other than $\hat{\omega}_{n}$ during any of the dispute rounds. If a dispute \textit{does} occur, and the dispute stake was below the threshold (that is, if $2^{n}d < M$), then another dispute round begins. And finally, if a dispute does occur, and if the dispute stake has met the threshold (that is, if $2^{n}d \ge M$), then the oracle resolves via the fork (just as with $\mathpzc{A}_{1}$), and those who staked on the winning outcome receive a pro rata share of all the stake that was staked on losing outcomes. The dispute sequence is formalized with the following definition.

\medskip
\noindent \textbf{Definition} (\textit{DisputeSequence}).
The algorithm $DisputeSequence$ accepts as input $(\mathtt{E},\Omega,\hat{\omega},\textbf{D})$, where $\mathtt{E}$ is an event, $\Omega$ is an outcomes space of $\mathtt{E}$, $\hat{\omega} \in \Omega$ is a tentative outcome, and $\textbf{D}=(D_{\omega})_{\omega \in \Omega \cup \left\{\texttt{Abstain} \right\} }$ is an $\Omega$-partition of the set of tokens that have been staked on some outcome $\omega \in \Omega$ during the present oracle query.

The algorithm runs a sequence of dispute rounds, terminating either when a dispute round completes without the tentative outcome being disputed, or until a dispute occurs for which the dispute stake is at least $M$. The algorithm returns $(\omega,\textbf{D},\texttt{FALSE}, \texttt{FALSE})$ if no dispute occurred in any dispute round. It returns $(\omega,\textbf{D},\texttt{TRUE}, \texttt{FALSE})$ if at least one dispute occurred, but no dispute occurred with dispute stake at least $M$. Finally, it returns $(\omega,\textbf{D},\texttt{TRUE}, \texttt{TRUE})$ if there was a dispute with dispute stake at least $M$. In pseudocode:

\begin{lstlisting}[language=json,escapeinside={*}{*}]
def *$DisputeSequence(\mathtt{E},\Omega,\hat{\omega}_{1},\textbf{D})$*:
  *$n \leftarrow 1 $*
  *$d \leftarrow |D_{\hat{\omega}}|$*
  *$\texttt{EVERDISPUTED} \leftarrow \texttt{FALSE}$*
  while *$2^{n-1}d < M$*
    //have a dispute round
    *$ (\hat{\omega}_{n+1},\textbf{D},\texttt{DISPUTED}) \leftarrow DisputeRound(\mathtt{E}, \Omega, \hat{\omega}_{n},\textbf{D},2^{n}d)$*
  
    //if there was no dispute
    if *$\texttt{DISPUTED} == \texttt{FALSE}$*:    
      return *$(\hat{\omega}_{n+1},\textbf{D}, \texttt{EVERDISPUTED}, \texttt{FALSE})$*
      
    else: // there was a dispute
      *$\texttt{EVERDISPUTED} \leftarrow \texttt{TRUE}$*
    endif
    
    *$n \leftarrow n+1$*
  endwhile
  
  return *$(\hat{\omega}_{n},\textbf{D},\texttt{TRUE}, \texttt{TRUE})$*
enddef
\end{lstlisting}
\medskip

Next we define the algorithm $BurnAndDistribute$, which behaves the same as $Distribute$ with the exception that it burns some of the dispute stake before distributing the rest, pro rata, to those that disputed in favor of the chosen outcome.

\medskip
\noindent \textbf{Definition} (\textit{BurnAndDistribute}).
The algorithm $BurnAndDistribute$ is called in the context of a dispute sequence. It accepts as input a tuple, $(\delta, \textbf{D}, \hat{\omega})$, where $\delta$ is a positive number (chosen by the oracle designer), $\textbf{D}$ is an $\Omega$-partition of dispute stake where $0 < \delta < |\bigcup \textbf{D}|$, and $\hat{\omega}$ is an outcome. The algorithm burns (by sending to a provably unspendable address) $\delta$ tokens from $\bigcup \textbf{D}$ and distributes the remaining tokens to the reporters, proportional to the number of tokens they staked in favor of $\hat{\omega}$.
\medskip

Next, we define a small variation of $ChoiceByFork$ that burns some amount of tokens before distributing the rest to token owners. Without such a burn, an attacker with access to a large amount of capital could cause forks without cost, by iteratively disputing every tentative outcome until a fork is initiated. (As long as the final tentative outcome is \texttt{True} the attacker would recoup -- as a reward for staking on \texttt{True} -- all the stake they lost by staking on false outcomes.)

\medskip
\noindent \textbf{Definition} (\textit{ChoiceByFork${}^{\prime}$}).
The algorithm $ChoiceByFork^{\prime}$ is a variation of $ChoiceByFork$ which simply calls $BurnAndDistribute(|\bigcup \textbf{D}| - \frac{7}{5}|D_{\hat{\omega}}|, \textbf{D},\hat{\omega})$ instead of $Distribute(\textbf{D}, \hat{\omega})$. Its purpose is to guarantee an ROI of 40\% to all token holders who disputed an outcome in favor of $\hat{\omega}$. (Note that, while we use a 40\% ROI in this paper, our results hold for any ROI strictly greater than 0\% and strictly less than 50\%.)
\medskip

With these definitions in place, we are ready to describe $\mathpzc{A}_{2}$, which behaves exactly the same as $\mathpzc{A}_{1}$ with the exception that $\mathpzc{A}_{2}$ uses $DisputeSequence$, $ChoiceByFork^{\prime}$, and $BurnAndDistribute$ in lieu of $DisputeRound$, $ChoiceByFork$, and $Distribute$. In pseudocode:

\begin{lstlisting}[language=json,escapeinside={*}{*}]
//initial state
*$C_{\texttt{True}}^{0} \leftarrow T_{genesis}$*
*$i \leftarrow 0$*
  
def *$\mathpzc{A}_{2}(\mathtt{E},\Omega,\phi,\hat{\omega},\{t_{1},\ldots,t_{d} \})$*:
  //increment query counter
  *$i \leftarrow i+1$*
  
  //update the reporting pool
  //only truth-tellers from previous query remain in pool
  *$T_{i} \leftarrow C_{\texttt{True}}^{i-1}$*
  
  //pay owners of tokens in *$T_{i}$*
  *$Pay(T_{i},\phi)$*
  
  //init *$\Omega$*-partition of dispute stake
  for *$\omega \in \Omega \cup \{\texttt{Abstain}\}$*:
    if *$\omega == \hat{\omega}$*:
      *$D_{\omega} \leftarrow \{t_{1},\ldots,t_{d} \}$*
    else:
      *$D_{\omega} \leftarrow \emptyset$*
    endif
    *$\textbf{D} \leftarrow \{ D_{\texttt{Abstain}}, D_{\omega_{1}}, \ldots, D_{\omega_{|\Omega|}} \}$*
  endfor
  
  //run the dispute sequence
  *$ (\hat{\omega}_{i},\textbf{D},\texttt{EVERDISPUTED}, \texttt{BIGDISPUTE}) \leftarrow DisputeSequence(\mathtt{E}, \Omega, \hat{\omega},\textbf{D})$*
  
  //if there was a large enough dispute  
  if *$\texttt{BIGDISPUTE} == \texttt{TRUE}$*:
    //resort to fork
    *$(\hat{\omega}_{i}, \{ C_{\texttt{Abstain}}^{i}, C_{\omega_{1}}^{i}, \ldots, C_{\omega_{|\Omega|}}^{i} \}) \leftarrow ChoiceByFork^{\prime}(\mathtt{E},\Omega, T_{i}, \textbf{D})$*
  
    //return outcome that won the fork
    return *$\hat{\omega}_{i}$*
  endif
  
  //if there was no dispute
  if *$\texttt{EVERDISPUTED} == \texttt{FALSE}$*:
    // Give the initial stake back to the querier
    *$Distribute(\textbf{D},\hat{\omega})$*
    
  else: //there was a dispute but not with dispute stake greater than *$M$*
    //burn some of the dispute stake*\footnote{Here, we are burning just enough dispute stake so that the ROI for those who disputed in favor of the \texttt{True} outcome is 40\%.}* and give the rest to the winners.
    *$BurnAndDistribute(|\bigcup\textbf{D}|-\frac{7}{5}|D_{\hat{\omega}_{i}}|,\textbf{D},\hat{\omega}_{i})$*
  endif
  
  //no reporting tokens are removed from the reporting pool
  *$C_{\texttt{True}}^{i}\leftarrow T_{i}$*
        
  //return tentative outcome
  return *$\hat{\omega}_{i}$*
enddef
\end{lstlisting}

\section{Analysis of $\mathpzc{A}_{2}$}

Observe that the only difference between $\mathpzc{A}_{1}$ and $\mathpzc{A}_{2}$ is that the latter runs $DisputeSequence$, $ChoiceByFork^{\prime}$, and $BurnAndDistribute$ whereas the former runs $DisputeRound$, $ChoiceByFork$, and $Distribute$. As before, we want to show that if the economic soundness condition is satisfied then, at equilibrium, the oracle $\mathpzc{A}_{2}$ is expected to return the \texttt{True} outcome of an oracle query \textit{without having to invoking a fork}. In particular, we want to show that, if honest play is expected during a fork, then there exists a unique subgame-perfect equilibrium in the subgame induced by the dispute sequence for which:

\begin{itemize}
\item The querier submits their query with the \texttt{True} outcome as the tentative outcome.
\item  When the tentative outcome for any dispute round is the \text{True} outcome, it is not disputed.
\item When a tentative outcome for any given dispute round is false, it is disputed in favor of the \texttt{True} outcome.
\end{itemize}

As we will show in the following theorem, when honest play is expected during forks, and when it is common knowledge that there exists at least one token holder who does not dispute a true tentative outcome in favor a false outcome, then there exists a unique subgame perfect equilibrium in the dispute sequence game that satisfies these three properties.

\begin{theorem}\label{thm:dispute-sequence-incentive-compatibility}
If honest play is expected during a fork, and if it is common knowledge among token holders that there exists at least one token holder who has not disputed in favor of a false outcome, and if the return for disputing false outcomes is chosen so that $0 < a < \frac{1}{2}$, then there exists a unique subgame perfect equilibrium in the dispute sequence game. The unique subgame perfect equilibrium results in the following behavior:
\begin{itemize}
\item The querier submits their query with the \texttt{True} outcome as the tentative outcome
\item If, during any dispute round, the \texttt{True} outcome is the tentative outcome, then no reporter is expected to dispute.
\item If, during any dispute round, a false outcome is the tentative outcome, then there exists a reporter whom we expect to dispute in favor of the \texttt{True} outcome.
\end{itemize}
\end{theorem}

\begin{proof}
See appendix.
\end{proof}

In addition to being incentive compatible and individually rational, this oracle can be made expensive to grief while remaining cheap to query. This is achieved via a judicious choice of the parameters $d$, $M$, and the proportion of tokens burned after each dispute. It is important to note, however, that as the quantity $M-d$ increases, so do the number of dispute rounds required to initiate a fork. Thus the maximum amount of time it can take for the oracle to respond to a query increases as the gap between $d$ and $M$ increases. Implementers of this oracle ought to keep this point in mind when choosing parameter values.

\section{Incentive Compatibility in the Cooperative Model}\label{section:cooperative_model}
Our main results show desirable properties in the non-cooperative model, but it is important to verify that these properties hold in the cooperative model, where players can make binding agreements -- after all, what is a smart contract if not a binding agreement? That is, we must verify that our results hold when users are able to collude and form coalitions. (This covers the case where the querier \textit{is} a reporter, because the utility of a ``querier-reporter" is the sum of the utilities of the querier role and the reporter role, and so a single ``querier-reporter" is equivalent, with respect to utility, to a coalition consisting of the querier and a separate reporter.) As we will see, the forking mechanism is trivially secure against collusion when the economic soundness condition is satisfied, and dispute rounds (and dispute sequences) are secure against collusion if the tokens in $T$ are sufficiently distributed.

\subsection{Coalitions in $\mathpzc{A}_{0}$}

\begin{figure*}
\centering
\includegraphics[width=5in]{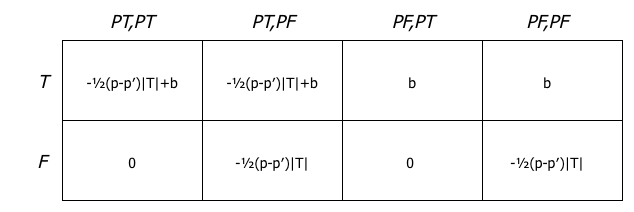}
\caption{The maximum total payout to a coalition that is both able to choose the outcome of an ($\mathpzc{A}_{0}$) oracle query and also includes the querier as a member. If the reporters cause the oracle to lie, their benefit is the querier's loss, so it is zero-sum: the benefits and losses to the coalition exactly cancel out.}
\label{fig:max_coalition_payouts}
\end{figure*}

First consider the simple oracle $\mathpzc{A}_{0}$ for which all oracle outputs are determined via the fork. Recall from section \ref{section:analysis_of_a_0} that the only strategy that is individually rational for the querier is to choose $PF,PF$, and when the querier chooses this strategy, any coalition of reporters large enough to unilaterally determine the output of the oracle has a unique best response: make the oracle return \texttt{True}. It follows that any coalition that is powerful enough to determine the outcome of the oracle query but which does not include the querier as a member does strictly better by making the oracle return the \texttt{True} outcome.

Next, consider a coalition that is both powerful enough to determine the outcome of the oracle and also includes the querier as a member. Such a coalition is able to unilaterally decide among the 8 outcomes in Figure \ref{fig:querier_decision_stage_game_normal_form}, and the payout to such a coalition is at most the sum of the payouts for the reporters and the querier for their chosen outcome. The maximum payouts for such a coalition are shown in Figure \ref{fig:max_coalition_payouts}. Note that the maximum payouts (assuming $b>0$) occur when the oracle returns the \texttt{True} outcome. Thus, any coalition that is powerful enough to decide the outcome of the oracle query does best when the oracle returns the \texttt{True} outcome.

When the economic soundness condition is satisfied, the analysis is even more straightforward. The economic soundness condition is satisfied exactly when the minimum cost of making the oracle return a false outcome is greater than the maximum total collective benefit for doing so. (Indeed, this is the entire motivation behind its definition.) If the economic soundness condition is satisfied, then the cost to any coalition that makes the oracle lie is greater than the maximum benefit that coalition could receive, and so by the pigeonhole principle the payout would not be individually rational for at least one member of the coalition. Therefore, any imputation must necessarily result in the oracle returning \texttt{True}.

\subsection{Coalitions in $\mathpzc{A}_{1}$ and $\mathpzc{A}_{2}$}

Next, consider oracles $\mathpzc{A}_{1}$ and $\mathpzc{A}_{2}$. For these oracles, there are two ways to get a false response to a query: either via a fork or by having a false tentative outcome go undisputed during a dispute round. For the reasons outlined in the previous section, we expect any coalition deciding the outcome of the oracle \textit{via a fork} to choose a behavior that results in the oracle returning the \texttt{True} outcome. So, for this section, we turn our attention to coalitions that may cause the oracle to return a false outcome by causing a false tentative outcome to go undisputed during a dispute round. We will first consider coalitions in $\mathpzc{A}_{1}$, and then show that the strategic situation for coalitions in $\mathpzc{A}_{2}$ can be reduced to those in $\mathpzc{A}_{1}$. We will assume, for this entire section, that the economic soundness conditions is satisfied.

For $\mathpzc{A}_{1}$, recall that during the dispute round (in the non-cooperative model) every token holder that controls at least $2d$ tokens in $T$ has the opportunity to dispute the tentative outcome. If they dispute the tentative outcome then they receive a benefit of $a2d$, and if they do not dispute the outcome they receive nothing.\footnote{In $\mathpzc{A}_{1}$ the value $a$ is $0.5$, while in $\mathpzc{A}_{2}$ the value $a$ is $0.4$.} In order for the oracle to return a false outcome to a valid oracle query without a fork, the querier must submit their query with a false tentative outcome, and all token holders (either individuals or mutually disjoint coalitions) must choose \textit{not} to dispute the false tentative outcome. In order for such behavior to be incentive compatible, each token holder would have to receive at least $a2d$ when choosing \textit{not} to dispute the false tentative outcome.

Therefore, any coalition payout that would result in the oracle $\mathpzc{A}_{1}$ returning a false outcome to a valid query without a fork would necessarily payout out at least $a2dn$, where $n$ is the number of players (either individuals or mutually disjoint coalitions) that control at least $2d$ tokens in $T$. In other words, the minimum cost of making $\mathpzc{A}_{1}$ return a false outcome without causing a fork is $a2dn$, because this is the total cost of bribing all other token holders to \textit{not} dispute the false tentative outcome.

As before, let $I$ denote the maximum (gross) benefit a coalition could receive by having the oracle return a false outcome to a valid oracle query. Then we can make the following simple observation relating the distribution of tokens in $T$ to the incentive compatibility of $\mathpzc{A}_{1}$ in the cooperative game-theoretic model: If $\frac{I}{a2d} < n$ then there does not exist any imputation for any coalition that results in the oracle returning a false outcome to a valid oracle query without calling the fork.

Colloquially, if the tokens in $T$ are sufficiently distributed, then the cost of paying all necessary token holders to \textit{not} exercise their opportunity to dispute the false tentative outcome is greater than the maximum benefit of doing so.

Finally, consider the oracle $\mathpzc{A}_{2}$. In this oracle there are several dispute rounds, each of which provides an opportunity for the oracle to return a false outcome without resorting to a fork. Analogous to the analysis of $\mathpzc{A}_{1}$, the $k$th dispute round requires a bond size of $2^{k} d$. Every player with at least $2^{k} d$ tokens in $T$ would have the opportunity to dispute a false tentative outcome during the $k$th dispute round. They would stand to gain $a2^{k} d$ for doing so, and would receive no benefit if they chose not to dispute. Thus, any individually rational payout for any coalition that would result in the oracle returning a false outcome in the $k$th dispute round (without calling a fork) would cost the coalition at least $a2^{k} dn$. Since such a coalition would receive a (gross) benefit at most $I$ for making the oracle return false, we make the following observation: If $\frac{I}{a2^{k} d} < n$ then there does not exist any imputation for any coalition that results in the oracle returning a false outcome to a valid oracle query without calling the fork.\footnote{Notice that, as $k$ grows, the token distribution burden is lessened. This is because the cost of bribing other players to \textit{not} exercise their ability to dispute a false tentative outcome grows exponentially with $k$.} This is maximized during the first dispute round, when the token distribution requirements for $\mathpzc{A}_{2}$ are identical to those in $\mathpzc{A}_{1}$.

In conclusion, the oracles in this paper can be expected to behave as intended in the cooperative model. For the oracle $\mathpzc{A}_{0}$ we need no further assumptions than those we made in the non-cooperative model. For oracles $\mathpzc{A}_{1}$ and $\mathpzc{A}_{2}$, incentive compatibility in the cooperative model requires that the tokens in the $T$ be sufficiently distributed.

\section{Conclusion}

We have introduced a new approach to decentralized oracle design -- one that is not based upon coordination games. We have presented three specific mechanisms which, under certain reasonable economic conditions, have been shown to be incentive compatible and individually rational in the non-cooperative model. Furthermore, we have shown that if the tokens in the reporting pool are sufficiently distributed, the mechanisms are also incentive compatible in the cooperative model.

\begin{acknowledgments}\label{section:acknowledgements}
The authors thank Micah Zoltu for his pioneering contributions to decentralized oracle design, and the Augur community for helpful feedback and discussions.  Financial support for this research was provided by Forecast Foundation OU.
\end{acknowledgments}

\bibliographystyle{unsrt}
\bibliography{main}

\appendix

\section{Calculations}\label{appendix:calculations}

\medskip
\noindent \textbf{Theorem \ref{th:if_sound_and_true_then_punish_false}.} \textit{If the economic soundness condition is satisfied then always choosing the move $PunishFalse$ is a best response by the querier to any strategy profile chosen by the reporters that results in the oracle returning $\texttt{True}$.}
\medskip

\begin{figure*}[ht]
\centering
\includegraphics[width=4in]{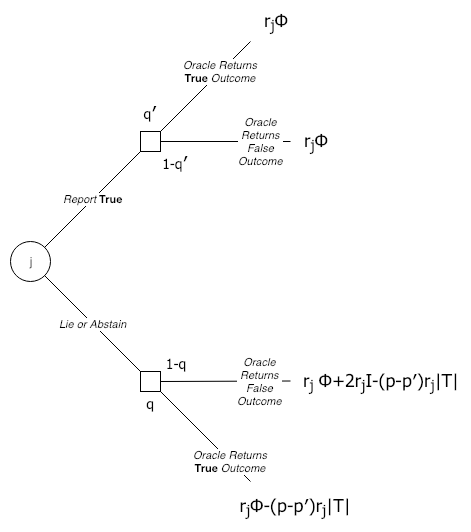}
\caption{A decision tree modeling the decision faced by an individual reporter in the stage game when the economic soundness condition is satisfied and the querier always chooses the move $PunishFalse$. The effects of the choices of the remaining reporters are modeled as chance moves. Figure \ref{fig:reporter_decision_stage_game_ev} shows a simplified version of this decision tree with the chance moves replaced by their expected values.}
\label{fig:reporter_decision_stage_game_with_chance}
\end{figure*}

\begin{figure}[b]
\centering
\includegraphics[width=4in]{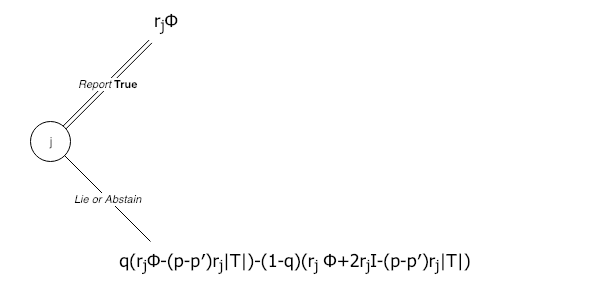}
\caption{A simplified version of the decision tree from Figure \ref{fig:reporter_decision_stage_game_with_chance}, with the chance moves replaced by their expected values. If the economic soundness condition is satisfied and the querier always chooses the move $PunishFalse$, then reporting the \texttt{True} outcome has a strictly better expected value than lying or abstaining.}
\label{fig:reporter_decision_stage_game_ev}
\end{figure}

\begin{figure*}[ht]
\centering
\includegraphics[width=5in]{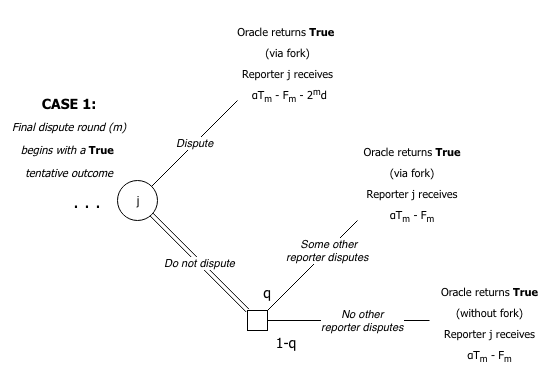}
\caption{A decision tree modeling the decision faced by an arbitrary reporter during a final dispute round in a dispute sequence in the case where the tentative outcome is \texttt{True}. The choice with the highest expected payouts is indicated by doubled lines.}
\label{fig:reporter_decision_final_dispute_round_case_1}
\end{figure*}

\begin{proof}
For the purposes of this proof, we will model the set of all reporters as a single player, referred to here as ``the reporter'', attempting to maximize its total payout. The reporter chooses whether to make the oracle return the \texttt{True} outcome or some false outcome. (In this way, we capture the set of all possible strategy profiles of \textit{actual} reporters and separate them into those that would cause the oracle to return the \texttt{True} outcome and those that would cause the oracle to return some false outcome; this simplifies our analysis significantly.) We assume that the reporter gets some benefit $I>0$ if the oracle returns a false outcome, and that this benefit comes at the expense of the querier. We further assume that the reporter will minimize their costs wherever possible, so that the cost of causing the oracle to lie (when the querier chooses to $PunishFalse$) is the minimum possible: $\frac{1}{2}(p-p^{\prime})|T|$. We make a similar conservative assumption for the cost of causing the oracle to return the $\texttt{True}$ outcome when the querier chooses to $PunishTrue$.\footnote{Observe that this assumption is extremely conservative. A coalition of reporters would have to know the querier's chosen strategy in advance in order to guarantee such low costs. For example suppose the reporter wanted to make the oracle return \texttt{True} while minimizing their cost of doing so. If the querier where going to choose $PunishTrue$, then the reporter would minimize costs by voting for \texttt{True} with only \textit{half} of the tokens in the reporting pool. If the querier where going to choose $PunishFalse$, then the reporter would minimize costs by voting for \texttt{True} with \textit{all} of the tokens in the reporting pool. Here we are being extremely conservative and are assuming that the reporter will always be able to minimize their costs no matter how the querier behaves.} We assume that the querier receives some benefit $b>0$ if and only if the oracle returns the \texttt{True} outcome. And finally, we note that the querier pays an oracle fee $\phi$ to the reporter no matter the final outcome.

Given these payouts, the resulting sequential game is show in extensive form in Figure \ref{fig:querier_decision_stage_game_extensive_form} and in normal form in Figure \ref{fig:querier_decision_stage_game_normal_form}.

When the economic soundness condition is satisfied, $I-\frac{1}{2}(p-p^{\prime})|T| < 0$, and so (as can be seen in Figures \ref{fig:querier_decision_stage_game_extensive_form} and \ref{fig:querier_decision_stage_game_normal_form}) the strategy profile $(True,(PunishFalse,PunishFalse))$ is a Pareto efficient, subgame-perfect Nash equilibrium. 

In other words, when the economic soundness condition is satisfied, always choosing the move $PunishFalse$ is a best response by the querier to any strategy profile of reporters that causes the oracle to return \texttt{True}.
\end{proof}

\medskip
\noindent \textbf{Theorem \ref{th:if_sound_and_punish_false_then_true}.} \textit{If the economic soundness condition is satisfied and the querier always chooses the move $PunishFalse$, then reporting the \texttt{True} outcome is always the best response by every individual reporter.}
\medskip

\begin{figure*}[ht]
\centering
\includegraphics[width=5in]{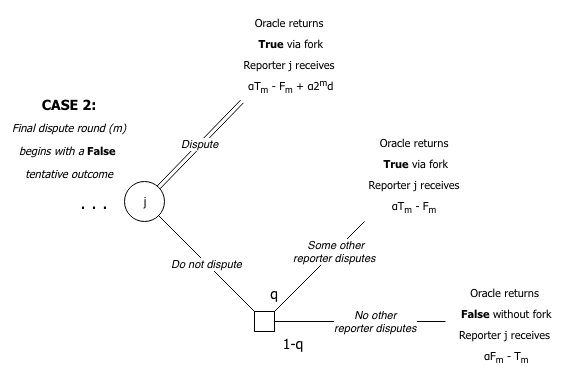}
\caption{A decision tree modeling the decision faced by an arbitrary reporter during a final dispute round in a dispute sequence in the case where the tentative outcome is false. By assumption, there exists at least one reporter for whom $F_m = 0$. For such a reporters, disputing is \textit{always} the dominant choice, regardless of the value of $q$. Indeed, since this fact is common knowledge among all reporters, all reporters know that $q=1$. Hence, in this case, disputing the false tentative outcome is the dominant choice for \textit{all} reporters, even if they hold large amounts of dispute stake on the false outcome.}
\label{fig:reporter_decision_final_dispute_round_case_2}
\end{figure*}

\begin{proof}
Suppose the economic soundness condition is satisfied and that the querier always chooses $PunishFalse$. Let $j$ be an arbitrary reporter. We will show that reporting \texttt{True} is the best response by $j$ no matter what choices are made by the remaining reporters.

Let $r_{j}$ denote the proportion of tokens in the reporting pool $T$ that are owned by $j$. Recall that each reporter always receives a pro rata share of the reporting fee $\phi$. In particular, reporter $j$ always receives $r_{j}\phi$.

Since the querier is always choosing the move $PunishFalse$, the reporter $j$ suffers a loss of $\frac{1}{2}(p-p^{\prime})r_{j}|T|$ if they lie or abstain during a fork.

We assume that if the oracle returns a false outcome then all reporters who lied or abstained will receive a pro rata share of $I$ (the total benefit of causing the oracle to lie). We make the conservative assumption that if the oracle is made to lie, then it was done at the minimum possible total cost to the set of lying reporters. (This is conservative because it maximizes the benefit to a lying reporter in the event that the oracle returns a false outcome.) That is, if the oracle returns a false outcome, then just $\frac{1}{2}|T|$ tokens were used to lie or abstain, and reporter $j$ will receive $2r_{j}I$ if and only if $j$ reported false or abstained during the fork.

The decision faced by reporter $j$ is modeled with the decision tree shown in Figure \ref{fig:reporter_decision_stage_game_with_chance} where the outcome of the oracle is modeled as a chance move. Replacing the chance moves with their expected values, we get the simplified decision tree in Figure \ref{fig:reporter_decision_stage_game_ev}, from which we can observe that the payoff to $j$ for reporting \texttt{True} is always $r_{j}\phi$ and the expected payoff for lying or abstaining is $q(r_{j}\phi-(p-p^{\prime})r_{j}|T|)+(1-q)(r_{j}\phi+2r_{j}I-(p-p^{\prime})r_{j}|T|)$. We need only show that the expected payoff for lying or abstaining is always strictly less than $r_{j}\phi$.

Because the economic soundness condition is satisfied, the quantities ${-(p-p^{\prime})r_{j}|T|}$ and $2r_{j}I-(p-p^{\prime})r_{j}|T|$ are both negative. It follows that the quantities $r_{j}\phi-(p-p^{\prime})r_{j}|T|$ and $r_{j}\phi+2r_{j}I-(p-p^{\prime})r_{j}|T|$ are both strictly less than $r_{j}\phi$. Thus the expected payoff for lying or abstaining is a convex combination of two values that are both strictly less than $r_{j}\phi$. Hence the expected payout for lying or abstaining is strictly less than $r_{j}\phi$.

Therefore, if the economic soundness condition is satisfied and the querier always chooses the move $PunishFalse$, then reporting the \texttt{True} outcome is always the best response by every individual reporter.
\end{proof}

\medskip
\noindent \textbf{Theorem \ref{thm:dispute-sequence-incentive-compatibility}.} \textit{If honest play is expected during a fork, and if it is common knowledge among token holders that there exists at least one token holder who has not disputed in favor of a false outcome, and if the return for disputing false outcomes is chosen so that $0 < a < \frac{1}{2}$, then there exists a unique subgame perfect equilibrium in the dispute sequence game. The unique subgame perfect equilibrium results in the following behavior:
\begin{itemize}
\item The querier submits their query with the \texttt{True} outcome as the tentative outcome
\item If, during any dispute round, the \texttt{True} outcome is the tentative outcome, then no reporter is expected to dispute.
\item If, during any dispute round, a false outcome is the tentative outcome, then there exists a reporter whom we expect to dispute in favor of the \text{True} outcome.
\end{itemize}}
\medskip

\begin{proof}
We will argue by backward induction, beginning with the final dispute round that would occur before a fork. We will show that in such a final dispute round, a \texttt{True} tentative outcome is expected to go undisputed, while a false tentative outcome is expected to be disputed in favor of the \texttt{True} outcome. This will form our base case.

Next, we will assume our induction hypothesis, which is that all dispute rounds after (and including) the $k$th dispute round are expected to have their tentative outcomes go undisputed if they are the \texttt{True} outcomes, and be disputed in favor of the \texttt{True} outcomes if they are false.

Then we will show that the induction hypothesis implies that we can expect the $(k-1)$th dispute round to have its tentative outcome go undisputed if it is the \texttt{True} outcome, and be disputed in favor of the \texttt{True} outcome if it is false. The conclusion is that we can always expect -- in every dispute round -- that the tentative outcome will go undisputed if it is the \texttt{True} outcome, and will be disputed in favor of the \texttt{True} outcome if it is false. An immediate consequence is that the querier is expected to submit their query with the \texttt{True} outcome as the initial tentative outcome.

\begin{figure*}[ht]
\centering
\includegraphics[width=5in]{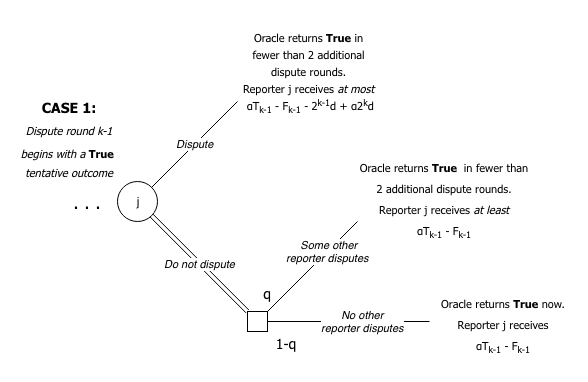}
\caption{A decision tree modeling the decision faced by an arbitrary reporter during the $(k-1)$th dispute round in the case where the tentative outcome is \texttt{True} and the induction hypothesis is assumed. The choice with the highest expected payouts is indicated by doubled lines.}
\label{fig:reporter_decision_k-1_dispute_round_case_1}
\end{figure*}

\textbf{Base Case:} Suppose the final dispute round in a dispute sequence occurs at round $m$. That is, if the tentative outcome of round $m$ is disputed, then $\mathpzc{A}_{2}$ will call the fork. There are two cases: either the tentative outcome at the beginning of the $m$th dispute round is the \texttt{True} outcome, or it is the false outcome.

\underline{Case 1:} \textit{The tentative outcome of the $m$\textup{th} dispute round is the \texttt{True} outcome.}

The strategic decision facing an arbitrary reporter, $j$, during the $m$th dispute round is shown in Figure \ref{fig:reporter_decision_final_dispute_round_case_1}, with $T_m$ denoting the total amount of dispute stake reporter $j$ has placed on \texttt{True} by the beginning of the $m$th dispute round, $F_m$ denoting the total amount of dispute stake reporter $j$ has placed on the false outcome by the beginning of the $m$th dispute round, $d$ denoting the amount of the querier's dispute stake, and $a$ denoting the ROI received by reporters for holding dispute stake on the outcome to which the oracle ultimately resolves. In our case, where each successive dispute round requires 2 times the dispute stake of the previous round and we burn $|\bigcup\textbf{D}|-\frac{7}{5}|D_{\hat{\omega}_{i}}|$ before distributing rewards, $a = 40\%$. (That is, if a reporter stakes $X$ tokens disputing in favor of \texttt{True}, and the oracle ultimately returns \texttt{True}, then the reporter will receive their original $X$ tokens back, in addition to $0.4X$ more.)

With this notation, disputing the tentative outcome requires a bond of size $2^m d$. Thus, if reporter $j$ disputes the \texttt{True} tentative outcome, the oracle will call the fork, which is expected to return \texttt{True}, which would result in reporter $j$ receiving a net payout of $aT_m - F_m - 2^m d$. Similarly, if the reporter does not dispute the tentative outcome, but someone else does, then the reporter can expect to receive a net payout of $aT_m - F_m$. Finally, if the reporter does not dispute the tentative outcome and nobody else does either, then the reporter will receive $aTm - Fm$.

Thus it is always best for a reporter to \textit{not} dispute the tentative outcome of the $m$th round if it is \texttt{True} -- no matter what other reporters choose to do.

\begin{figure*}[ht]
\centering
\includegraphics[width=5in]{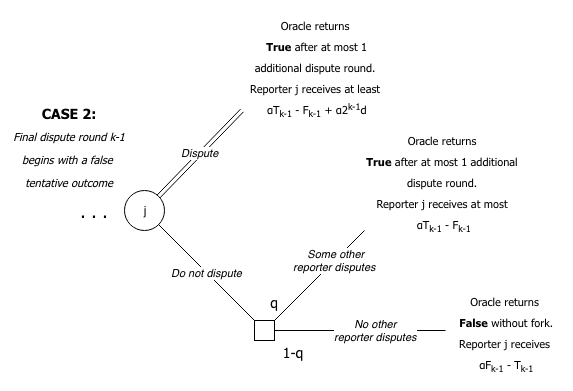}
\caption{A decision tree modeling the decision faced by an arbitrary reporter during the $(k-1)$th dispute round in the case where the tentative outcome is false and the induction hypothesis is assumed. The choice with the highest expected payout is indicated by doubled lines.}
\label{fig:reporter_decision_k-1_dispute_round_case_2}
\end{figure*}

\underline{Case 2:} \textit{The tentative outcome of the $m$\textup{th} dispute round is a false outcome.}

The strategic decision facing an arbitrary reporter, $j$, during the $m$th dispute round is shown in Figure \ref{fig:reporter_decision_final_dispute_round_case_2}, using the same notation as in case 1. In this case, if reporter $j$ decides to dispute, she can expect a payout of $aT_m -F_m + a2^m d$. Her expected payout if she doesn't dispute is $q(aT_m - F_m)+(1-q)(aF_m - T_m)$ where $0 \le q \le 1$.

Recall that, by assumption, there exists at least one token holder for whom $F_m = 0$. For such a token holder, the expected payout for disputing is $aT_m + a2^m d$, while the expected payout for not disputing is $qaT_m - (1-q)(T_m)$.  Thus, for this disputer, choosing to dispute is always a dominant strategy, no matter what other reporters do, because for all $q$ where $0 \le q \le 1$, $aT_m + a2^m d$ is greater than $qaT_m - (1-q)(T_m)$. So, we can always expect the false tentative outcome to be disputed.\footnote{Despite first appearances, we are not resting our fate on the hopes that just a \textit{single} reporter will be motivated to dispute a false tentative outcome. The situation is not so dire. Observe that, since the existence of this reporter is common knowledge, all reporters expect that $q=1$.  Hence, in this case, disputing the false tentative outcome is also the dominant choice -- by quite a large margin -- for \textit{all} token holders, even if they hold the maximum possible amount of false dispute stake.}

This concludes our base case. We have shown that in a final dispute round, a \texttt{True} tentative outcome is expected to go undisputed, while a false tentative outcome is expected to be disputed in favor of the \texttt{True} outcome.

\textbf{Induction Step:} 
Now suppose, as our induction hypothesis, that for all $k$ where $2 \le k \le m$, we can expect that in the $k$th dispute round a \texttt{True} tentative outcome will go undisputed, while a false tentative outcome will be be disputed in favor of the \texttt{True} outcome. We need to show that we can expect the same behavior in the $(k-1)$th dispute round. As before we have two cases: either the tentative outcome in the $(k-1)$th round is the \texttt{True} outcome, or it is a false outcome.

\underline{Case 1:} \textit{The tentative outcome of the $(k-1)$th dispute round is the \texttt{True} outcome.}
The strategic decision facing an arbitrary reporter, $j$, during the $(k-1)$th dispute round in the case where the tentative outcome is \texttt{True} is shown in Figure \ref{fig:reporter_decision_k-1_dispute_round_case_1}. Choosing \textit{not} to dispute the \texttt{True} tentative outcome results in the reporter receiving a payout of at least $aT_{k-1} - F_{k-1}$. Choosing to dispute the \textit{True} tentative outcome gives the reporter a payout of at most $aT_{k-1} - F_{k-1} - 2^{k-1} d + a2^{k} d$. When $0 < a < \frac{1}{2}$, it is the case that $- 2^{k-1} d + a2^{k} d < 0$, so choosing \textit{not} to dispute the \texttt{True} outcome is the strictly dominant decision, no matter what other reporters choose to do.

\underline{Case 2:} \textit{The tentative outcome of $(k-1)$th dispute round is a false outcome.}
The strategic decision facing an arbitrary reporter, $j$, during the $(k-1)$th dispute round in the case where the tentative outcome is false is shown in Figure \ref{fig:reporter_decision_k-1_dispute_round_case_2}. Recall that, by assumption, there exists at least one token holder for whom $F_m = 0$. For such a token holder, the expected payout for disputing the false tentative outcome in favor of the \texttt{True} outcome is at least $aT_{k-1} + a2^{k-1} d$, while the expected payout for not disputing is at most $aT_{k-1}$. Thus, for this disputer, choosing to dispute the false tentative outcome in favor of the \texttt{True} outcome is always a strictly dominant strategy, no matter what other reporters do. Hence we can expect that the false tentative outcome will be disputed in favor of the \texttt{True} outcome.\footnote{As in the base case, we are not resting our fate on the hopes that just a \textit{single} reporter will be motivated to dispute a false tentative outcome. Since the existence of this reporter is common knowledge, all reporters expect that $q=1$.}

This completes the induction step and shows that in every dispute round of the dispute sequence we can expect the tentative outcome to be disputed if and only if it is a false outcome. It follows immediately that we can expect the querier to submit their query with the \texttt{True} outcome as the initial tentative outcome.
\end{proof}

\end{document}